%% file: irregular_mgle-arxiv-v2.tex
\documentclass[a4paper, 11pt]{article}

\usepackage{amssymb,amsmath,amsthm,lmodern,cancel}
\usepackage{amsfonts}
\usepackage[activate={true,nocompatibility},final,tracking=true,kerning=true,spacing=true,factor=1100,stretch=15,shrink=15]{microtype}
\microtypecontext{spacing=nonfrench}

\usepackage{graphicx}
\usepackage[english]{babel}
\usepackage[T1]{fontenc}
\usepackage{textcomp}
\usepackage[dvipsnames]{xcolor}
\usepackage[shortlabels]{enumitem}
\usepackage{thmtools, mathtools,verbatim} \usepackage{upgreek}
\usepackage{mathrsfs}
\usepackage{accents}
\usepackage{bbold}

\usepackage[colorlinks=false,hyperindex=true]{hyperref}
\usepackage{cleveref}

\usepackage[font=normal]{caption}

\usepackage{thmtools}
\usepackage{thm-restate}

\usepackage{amssymb}

\usepackage{ifthen}

\usepackage{tikz}
\usetikzlibrary{math}
\usetikzlibrary{arrows, positioning}
\tikzset{
    >=stealth',
    punkt/.style={
           rectangle,
           rounded corners,
           draw=black, very thick,
           text width=6.5em,
           minimum height=2em,
           text centered},
    pil/.style={
           ->,
           thick,
           shorten <=2pt,
           shorten >=2pt,}
}

\usepackage[bottom]{footmisc}

\title{Irregular SLE$_{4}$ martingales and isomonodromic deformations}  
\date{}

\author{Harini Desiraju\thanks{Mathematical Institute, University of Oxford, United Kingdom. \protect\url{harini.desiraju@maths.ox.ac.uk}}, 
\, 
Aleksandra Korzhenkova\thanks{Institute of Mathematics, 
EPFL, Switzerland. \protect\url{aleksandra.korzhenkova@gmail.com}}, 
\, and \,
Eveliina Peltola\thanks{Department of Mathematics and Systems Analysis, Aalto University, Finland; and \\ Division of Mathematics, University of Cologne, Germany.   \protect\url{eveliina.peltola@aalto.fi}}
}

\setlist[enumerate]{topsep = 1ex, leftmargin=1cm, itemsep= -2pt}

\let\OLDthebibliography\thebibliography
\renewcommand\thebibliography[1]{
\OLDthebibliography{#1}
\setlength{\parskip}{1pt}
\setlength{\itemsep}{2pt}
}

\allowdisplaybreaks

\newtheorem{thm}{Theorem}[section]
\crefname{thm}{theorem}{theorems}

\newtheorem{cor}[thm]{Corollary}
\crefname{cor}{corollary}{corollaries}

\crefname{conj}{conjecture}{conjectures}

\newtheorem{lem}[thm]{Lemma}
\crefname{lem}{lemma}{lemmas}

\newtheorem{prop}[thm]{Proposition}
\crefname{prop}{proposition}{propositions}

\crefname{mainresult}{result}{results}

\crefname{theorem}{theorem}{theorems}

\crefname{lemA}{lemma}{lemmas}

\theoremstyle{definition}

\newtheorem{remark}[thm]{Remark}

\numberwithin{equation}{section}
\numberwithin{figure}{section}

\input{tex-arXiv-v2/macros}

\begin{document}
\maketitle
\begin{abstract}
We consider non-Fuchsian monodromy preserving deformations on a Riemann sphere. The associated isomonodromic deformation parameters on this surface comprise the positions of the singularities, together with the Birkhoff (spectral) invariants owing to the presence of irregular singularities. Our first main result is the derivation of the Loewner evolution of these isomonodromic deformation parameters. Using this result, we construct martingale observables for Schramm-Loewner evolution ($\SLE_4$) processes in the presence of double poles. Geometrically, the expressions contain the pre-Schwarzian and Schwarzian of the Loewner evolution, arising from conformal covariance of the observable. Furthermore, we characterize these $\SLE_4$ observables uniquely in terms of confluent BPZ equations of a CFT with central charge $c=1$.
\end{abstract}

\setcounter{tocdepth}{2}
\tableofcontents

\newpage

\newpage


\section{Introduction}
\input{tex-arXiv-v2/1-intro}


\section{Loewner evolution of Birkhoff invariants and confluence}
\label{sec:evol_of_parameters}
\input{tex-arXiv-v2/2-stoch}


\section{Irregular $\SLE_4$ martingales}
\label{sec:mgle}
\input{tex-arXiv-v2/3-SLEm2}

\bibliographystyle{alpha}
\bibliography{jabref_database-2025}
\end{document}

%% file: tex-arXiv-v2/macros.tex
\global\long\def\bR{\mathbb{R}}

\global\long\def\bC{\mathbb{C}}
\global\long\def\bH{\mathbb{H}}

\global\long\def\cX{\mathcal{X}}

\global\long\def\cC{\mathcal{C}}

\global\long\def\cF{\mathcal{F}}
\global\long\def\cG{\mathcal{G}}
\global\long\def\cU{\mathcal{U}}

\global\long\def\cXFP{\mathcal{X}}
\global\long\def\FunSp{\cC^2_{\mathrm{hol}}}
\global\long\def\FunSpGen{\cC^{2,1}_{\mathrm{hol}}}

\newcommand{\bs}[1]{{\boldsymbol{#1}}}

\global\long\def\PartF{\mathscr{Z}}

\global\long\def\ii{\mathrm{i}}

\newcommand{\res}{\mathrm{Res}}
\newcommand{\Tr}{\mathrm{Tr}}

\global\long\def\ud{\,\mathrm{d}}

\renewcommand\Re{\operatorname{Re}}
\renewcommand\Im{\operatorname{Im}}

\global\long\def\GL{\operatorname{GL}}
\global\long\def\SL{\operatorname{SL}}
\global\long\def\sl{\mathfrak{sl}}

\DeclareMathAlphabet{\mathpzc}{OT1}{pzc}{m}{it}

\global\long\def\one{\scalebox{1.1}{\textnormal{1}} \hspace*{-.75mm} \scalebox{0.7}{\raisebox{.3em}{\bf |}} }

\newcommand{\cond}{\, | \,}

\newcommand{\SLE}{\operatorname{SLE}}
\newcommand{\CLE}{\operatorname{CLE}}

\newcommand{\Driv}{Z}
\newcommand{\Spec}{\Xi}
\newcommand{\spec}{\xi}

\makeatletter
\newif\iflibus@sansmath
\makeatother
\DeclareFontEncoding{LS1}{}{}
\DeclareFontSubstitution{LS1}{libertinust1math}{m}{n}
\DeclareSymbolFont{LettersLibertinus}{LS1}{libertinust1math}{m}{it}

\newcommand{\Schwarzian}{\mathscr{S}}
\newcommand{\PreSchwarzian}{\mathscr{A}}

%% file: tex-arXiv-v2/1-intro.tex
Schramm-Lowner evolutions ($\SLE_{\kappa})_{\kappa \geq 0}$ are universal conformally invariant random curves, which have played a central role in probability theory and statistical physics; 
in complex geometry and conformal field theory; having striking further connections to various other areas of mathematics and theoretical physics:
Teichm\"uller theory, enumerative geometry, integrable systems, random matrices, quantum gravity, and beyond. 
Our motivation in this work is to shed light on the intrinsic connections of $\SLE$ random curves with geometry.

In~\cite{Dubedat:Double_dimers_conformal_loop_ensembles_and_isomonodromic_deformations}, 
Dub\'edat established a relationship between certain $\SLE_\kappa$ martingale observables with $\kappa=4$ and 
isomonodromic tau-functions with simple poles associated to an $\SL_2(\bC)$ linear system on a punctured sphere. 
Basok~\&~Chelkak further investigated the matter in~\cite{Basok-Chelkak:Tau_functions_a_la_Dubedat_and_probabilities_of_cylindrical_events_for_double_dimers_and_CLE4}, 
revealing the more precise connection to Fock-Goncharov lamination basis of the character variety, pointing towards a more concrete link to Teichm\"uller theory.
In this spirit, the key focus of the present article is the relation between tau-functions and isomonodromic deformations with SLE martingales, 
inspired by applications towards topological observables for random curves, which we also plan to return to in future work.

We consider $\SL_2(\bC)$ \emph{non-Fuchsian} linear systems, where the tau-functions
depend on both the locations of the (irregular) singularities and of the so-called Birkhoff (spectral) invariants.  
We derive the Loewner evolution of these spectral invariants on the Riemann sphere (Theorems~\ref{thm:Sto-Bir-Gen}~\&~\ref{thm:evol_isomonodromic_times}).
 It is noteworthy that the analytic structure of this evolution is independent of the rank of the singularity.
We then aim to characterize $\SLE_4$ martingale observables in this more general setup, involving higher geometric invariants. 
First of all, we construct explicit martingale observables in the case of double poles (\Cref{thm:main_mgle}),
and second, we show that they are uniquely characterized as solutions to linear hypoelliptic differential equations, 
dubbed ``confluent'' Belavin-Polyakov-Zamolodchikov (BPZ) equations (Theorems~\ref{thm:BPZ_general}~\&~\ref{thm:non_BPZ_PDE2}).
Geometrically, the expressions contain covariance factors involving the pre-Schwarzian and Schwarzian of the Loewner evolution.
We also discuss how to extend these results to the case of higher order singularities, which will involve higher Schwarzians as well --- suggesting a richer underlying geometric structure.

Our work represents a first step towards establishing the geometric role of conformally invariant random curves on Riemann surfaces with irregular singularities.
For example, based on our construction, it will be very interesting to explore what the geometric and asymptotic data means for the SLE curves and dimer models. 
One might discover new stochastic phenomena which could arise from the effects of the presence of irregular singularities and non-trivial topological data, or 
the Stokes phenomenon~\cite{Stokes:On_the_discontinuity_of_arbitrary_constants_which_appear_in_divergent_developments,Ablowitz-Fokas:Complex_variables}. 
This perspective may also lead to the identification of new observables in lattice models, in light of the already known connection of $\SLE_4$ curves with the double-dimer model~\cite{Dubedat:Double_dimers_conformal_loop_ensembles_and_isomonodromic_deformations, Basok-Chelkak:Tau_functions_a_la_Dubedat_and_probabilities_of_cylindrical_events_for_double_dimers_and_CLE4}.

\subsection{Schramm-Loewner evolution and stochastic flow on moduli space}

\begin{figure}[h!]
\centering
\includegraphics[width=0.4\textwidth]{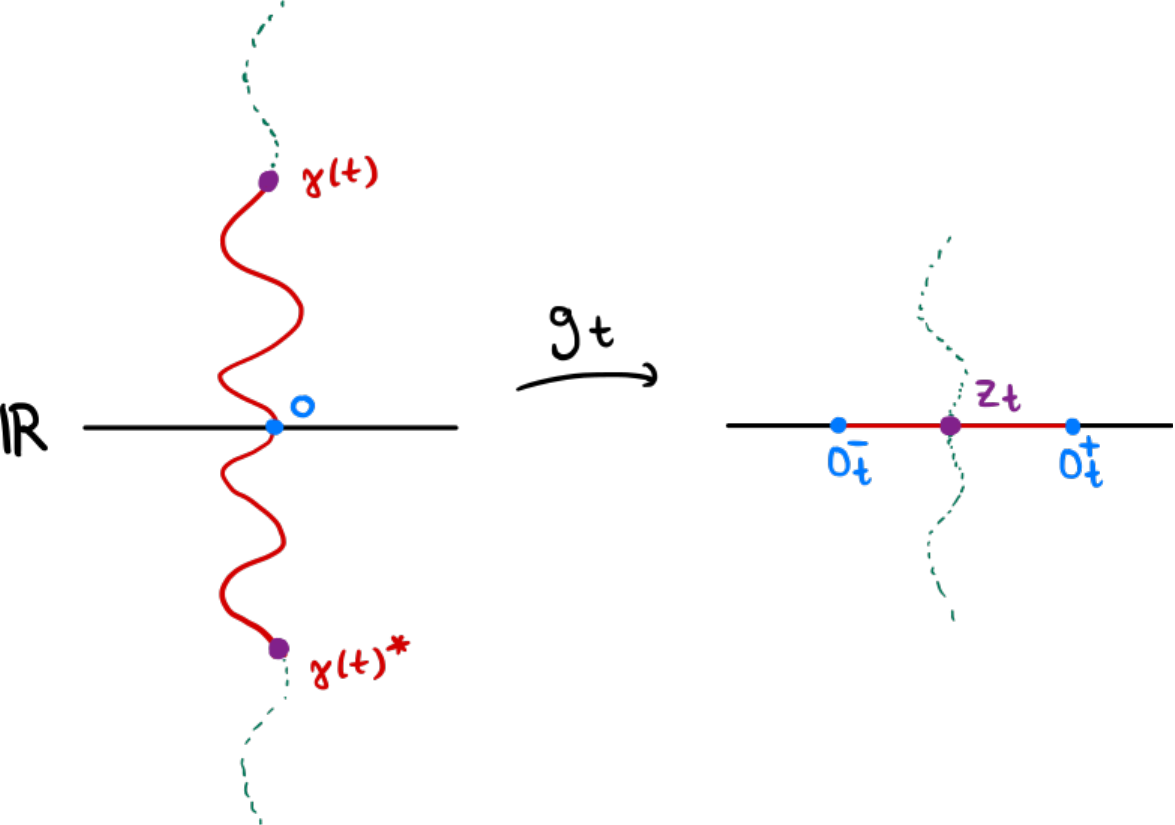}
\qquad\qquad
\includegraphics[width=0.35\textwidth]{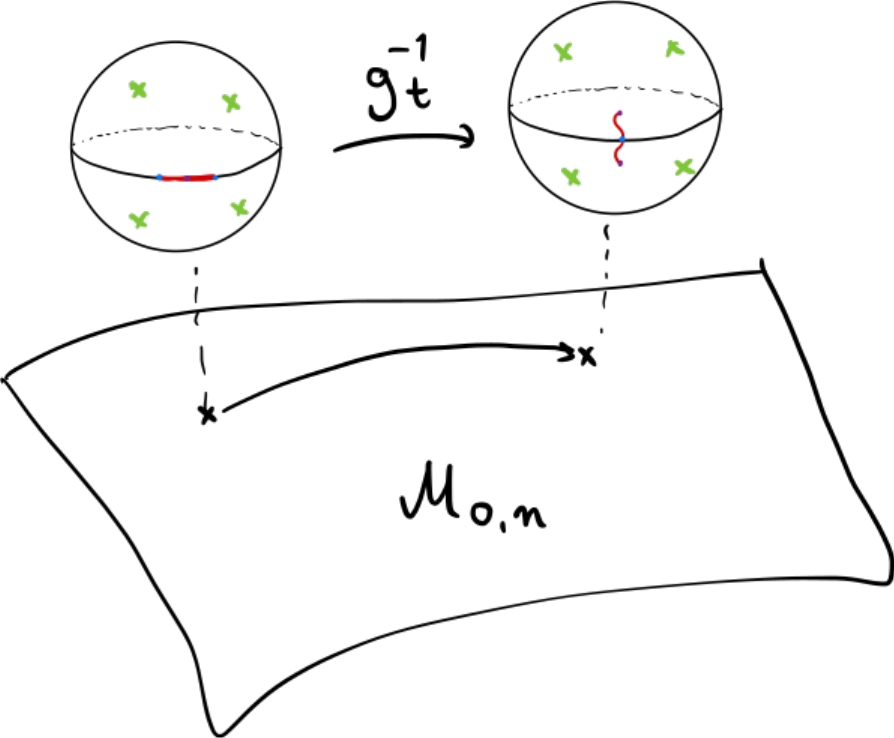}
\caption{\label{fig:SLE}
{Illustration of a chordal injective curve $\gamma \colon [0,\infty) \to \bH$ in the upper half-plane $\bH := \{ z \in \bC \cond \Im(z)>0 \}$ connecting the boundary points $0 = \gamma(0)$ and $\infty = {\underset{t \to \infty}{\lim} \, \gamma(t)}$. \\ 
The \emph{Loewner flow} on the Riemann sphere $\hat\bC \coloneqq \bC \cup \{\infty\}$ is obtained from conformal bijections $g_t \colon \hat\bC \setminus (\gamma[0,t] \cup \gamma^*[0,t]) \to \hat\bC \setminus [O^-_t,O^+_t]$, where $O^\pm_t=g_t(0\pm) \in \bR$ are the two images of $0$ under the flow.
The Loewner \emph{driving function} of the evolution is $\Driv_t = g_t(\gamma(t)) = g_t(\gamma^*(t))$
and the tip of the curve can be recovered from the radial limit
$\gamma(t) = {\underset{\varepsilon \to 0+}{\lim} \, g_t^{-1}(\Driv_t + \ii \varepsilon)}$. \\
Roughly speaking, at each time instant $t$, the map $g_t$ deforms the metric on $\bH$ by pullback: $g_t^{-1}$ can be viewed as a map on $\hat\bC$, conformal on $\hat\bC \setminus [O^-_t,O^+_t]$ and having radial limits on $[O^-_t,O^+_t]$, deforming the conformal class of the metric on that interval. 
By keeping track of marked points off the curve, we can view the Loewner flow of the curve as a flow on the moduli space of punctured Riemann spheres. 
If the driving function $t \mapsto \Driv_t \in \bR$ is random (e.g. Brownian motion), we obtain a stochastic Loewner flow (like for SLE).
} $\heartsuit$}
\end{figure}

It has become clear during the past couple of decades that Schramm's SLE random curves~\cite{Schramm:ICM},  
which were originally introduced to describe critical interfaces in statistical physics, 
share a deep relation with geometry and integrability, and with conformal field theory (CFT). 
SLEs can be thought of as random curves in two-dimensional space, and they can be dynamically generated from evolution of conformal mappings using classical Loewner theory.  
As key input, one needs a random one-dimensional process, termed driving process, 
which describes the time-evolution of the tip of the growing curve 
via Loewner's differential equation~\eqref{eq:LE} (cf.~\cite{Loewner:Untersuchungen_uber_schlichte_konforme_Abbildungen_des_Einheitskreises}). 
This description makes SLE curves concrete to work with by utilizing standard tools from stochastic analysis and function theory. 

These curves are universal: Schramm showed in~\cite{Schramm:Scaling_limits_of_LERW_and_UST} 
that the one-parameter family $(\SLE_\kappa)_{\kappa > 0}$ of random chordal curves in a simply connected domain 
is the only possible one which is 
\emph{conformally invariant} (in the sense that the distribution of the random curve in any simply connected domain is obtained from the one on the upper half-plane $\bH$ by pushforward of any uniformizing map) and \emph{Markovian} 
(in the sense that it has the ``domain Markov property'' with respect to the growth of the curve\footnote{Conditional on an initial segment of the curve, the distribution of the remaining piece of the curve is that of the $\SLE_\kappa$ but in the slit domain obtained from the original domain by removing the initial segment. (Strictly speaking, when $\kappa > 4$, the $\SLE_\kappa$ curve is not injective, but it has self-touchings. In this case its complementary domain has several components, and one should only uniformize the one containing the intended target point of the curve. We omit this detail from our discussion and work with $\kappa = 4$ only.)}).
Namely, imposing these two properties dictates that the Loewner driving process of the random curve 
must be a multiple of standard one-dimensional \emph{Brownian motion} --- the multiplicative constant being the variance of the driving process, $\sqrt{\kappa}$ for some $\kappa>0$.
We will describe Loewner theory and $\SLE_\kappa$ curves in some more detail in \Cref{subsec:preli_SLE}; see also \Cref{fig:SLE}. 

\medskip

Geometrically, the evolution of an SLE curve can be seen as a \emph{stochastic flow} (a diffusion process) on the associated moduli space of Riemann surfaces 
(see, e.g.,~\cite{Bauer-Bernard:SLE_growth_processes_and_CFTs, Bauer-Bernard:Conformal_field_theories_of_SLEs, Friedrich-Werner:Conformal_restriction_highest_weight_representations_and_SLE, Kontsevich:CFT_SLE_and_phase_boundaries, Friedrich-Kalkkinen:On_CFT_and_SLE, Friedrich:On_connections_of_CFT_and_SLE, Kontsevich-Suhov:On_Malliavin_measures_SLE_and_CFT, Dubedat:SLE_and_Virasoro_representations_localization}).
In the case of simply connected domains\footnote{A similar picture should remain for the case of general Riemann surfaces, although the problem becomes significantly more involved due to the additional conformal moduli. 
See~\cite{Zhan:Thesis, Kontsevich-Suhov:On_Malliavin_measures_SLE_and_CFT, Dubedat:SLE_and_Virasoro_representations_localization}.}, 
which is one of the few cases where the SLE evolution is rigorously understood, 
the shrinking slit domains generated by the SLE curve represent different elements of the moduli space of the upper half-plane, say, meaning that the metric changes upon deformation by the growing curve.  
After uniformizing the slit domain back to the upper half-plane, one obtains an interpretation of the curve evolution in the moduli space.

Bauer~\&~Bernard observed early on that conserved quantities for the SLE flow, encoded in \emph{martingale observables}, in fact generate special representations of the Virasoro algebra~\cite{Bauer-Bernard:SLE_martingales_and_Virasoro_algebra, Bauer-Bernard:Conformal_transformations_and_SLE_partition_function_martingale}. 
This is not completely surprising, since the Virasoro algebra is the infinite-dimensional Lie algebra that provides the conformal symmetry for critical models and CFT --- and SLE processes are manifestly conformally invariant. 
In contrast, what was rather remarkable is the fact that points where SLE curves emerge are described by CFT primary fields which play a very special role in the representation theory of the Virasoro algebra.
Namely, they correspond to \emph{degenerate fields}, whose correlation functions solve linear homogeneous partial differential equations, termed 
BPZ PDEs~\cite{BPZ:Infinite_conformal_symmetry_in_2D_QFT}.

From a completely different route, one can observe the very same PDEs to emerge from \emph{stochastic calculus}, by requiring natural SLE observables to be martingales for the time-evolution of the curve. 
In the language of probability theory, the 2nd order BPZ PDEs satisfied by correlation functions of these degenerate fields are equivalent to the vanishing of the It\^o drifts in the associated martingale observables. 
(See~\cite{Peltola:Towards_CFT_for_SLEs} and references therein.) 
We briefly introduce the stochastic analysis underlying SLE theory in \Cref{subsec:preli_SLE}, and delve into martingales in \Cref{subsec:mgle_double_pole,subsec:mgle_higher_poles}, 
which also contain one of our main results (\Cref{thm:main_mgle}). 
In \Cref{subsec:BPZ}, we discuss the confluent BPZ PDEs for the martingale observables at the presence of irregular singularities (Theorems~\ref{thm:BPZ_general} and~\ref{thm:non_BPZ_PDE2}).

\subsection{Isomonodromic defomations and tau-functions}

The theory of moduli spaces of Riemann surfaces is intricately connected to integrable systems and 
to ordinary differential equations (ODEs) dubbed \emph{isomonodromic deformation equations} (IDEs)~\cite{BBT:Introduction_to_classical_integrable_systems, FIKN:Painleve_transcendents_Riemann-Hilbert_approach}. 
Consider the Fuchsian matrix ODE
\begin{align}\label{eq:Fuchsian}
\frac{\partial}{\partial z} Y(z) = \sum_{i=1}^{n} \frac{A_{i}}{z-\lambda_i} \, Y(z), 
\qquad A_i \in \sl_2(\bC) , 
\end{align}
for $z\in \hat\bC_{\bs\lambda} \coloneqq \hat\bC \setminus \{\lambda_1,\ldots,\lambda_n\}$ on the Riemann sphere $\hat\bC \coloneqq \bC \cup \{\infty\}$ with $n$ punctures (singularities) $\bs\lambda=(\lambda_1,\ldots,\lambda_n)$,
assuming that $A_i$ are diagonalizable with eigenvalues $(\alpha_{i},-\alpha_{i})$ with $\Re(\alpha_{i}) \geq 0$. 
The monodromy of its local solution $Y$ then gives rise to a representation of the fundamental group of $\hat\bC_{\bs\lambda}$.
Furthermore, one may consider the stability of the monodromy: which deformations of the punctures $\bs\lambda$ on $\hat\bC$ preserve the monodromy representation of the function $Y(z)$?
The answer in the above classical case is provided by the condition
\begin{align}
\frac{\partial}{\partial \lambda_i} Y(z,\bs\lambda) = - \frac{A_{i}(\bs\lambda)}{z-\lambda_i} \, Y(z,\bs\lambda), \quad \textnormal{for all } i \in\{1,\ldots,n\} ,\label{eq:Fuchsian_IDE}
\end{align}
or equivalently, that the matrices $A_i(\bs \lambda)$ solve the Schlesinger equations~\cite{Schlesinger:Uber_eine_Klasse_von_Differentialsystemen_beliebiger_Ordnung_mit_festen_kritischen_Punkte}
\begin{align*}
\frac{\partial A_i}{\partial \lambda_j} = \frac{\left[A_i, A_j \right]}{\lambda_i - \lambda_j},
\qquad 
\frac{\partial A_i}{\partial \lambda_i} =- \sum_{i\neq j} \frac{\left[A_i, A_j \right]}{\lambda_i - \lambda_j}.
\end{align*}

An important quantity in this context is the so-called \emph{tau-function} $\tau(\bs\lambda)$, 
which for the case of the IDEs~\eqref{eq:Fuchsian_IDE} is defined via Hamiltonians $H_{\lambda_i}$ associated to the punctures, 
\begin{align*}
\Big( \frac{\partial}{\partial \lambda_i} \log \tau(\bs\lambda) \Big) \ud\lambda_i = H_{\lambda_i} , 
\qquad \textnormal{where}
\quad H_{\lambda_i}\coloneqq \frac{1}{2}  \, \underset{z=\lambda_i}{\res}\, \Tr \Bigg(\sum_{i=1}^n\frac{A_i(\bs\lambda)}{z-\lambda_i} \Bigg)^2 \ud z , \quad \textnormal{for all } i ;
\end{align*}
see \cite{Balogh-Harnad:Tau_functions_and_their_applications} for an overview. 
Geometrically, tau-functions can be viewed as sections of line bundles over the moduli space of punctured Riemann 
surfaces\footnote{Geometrically speaking, Schlesinger's equations~\eqref{eq:Fuchsian_IDE} describe flat holomorphic $\sl_2(\bC)$-connections on the $n$-punctured Riemann sphere. 
One can extend such a formalism to flat connections on Riemann surfaces with varied punctures (and one could also consider other Lie algebras for the coefficients).
In particular, for the case $n=4$ of four-punctured Riemann surfaces, the flat $\sl_2(\bC)$-connections give rise to the Lax equations for the Painlev\'e VI equation~\cite{Painleve:Sur_les_equations_differentielles_du_second_ordre_a_points_critiques_fixes, 
Malmquist:Sur_les_equations_differentielles_du_second_ordre_dont_lintegrale_generale_a_ses_points_critiques_fixes, 
Okamoto:On_the_tau-function_of_Painleve_equations, JMU:Monodromy_preserving_deformation_of_linear_ODEs_with_rational_coefficients1}.}. 
They describe a distinguished class of nonlinear special functions that appear in several other areas of mathematical physics. 
They arise in the study of correlation functions in statistical physics and quantum field theory~\cite{SMJ:Holonomic_quantum_fields3},  
they were found to give expressions for conformal blocks in CFT (see in particular~\cite{ILT:Isomonodromic_tau-functions_from_Liouville_conformal_blocks})  
and partition functions of random matrix models (see in particular~\cite{BEH:Partition_functions_for_matrix_models_and_isomonodromic_tau_functions}) 
--- and they have important recent applications in statistical physics and 
the geometry of random conformally invariant curves~\cite{Dubedat:Double_dimers_conformal_loop_ensembles_and_isomonodromic_deformations, Basok-Chelkak:Tau_functions_a_la_Dubedat_and_probabilities_of_cylindrical_events_for_double_dimers_and_CLE4} (\Cref{subsec:mgle_intro}).

From the integrability and geometry point of view, the above Fuchsian case is very classical. In the 1980s, 
Jimbo, Miwa, and Ueno systematically extended this picture to the non-Fuchsian case, which involves higher order poles (irregular singularities), 
in part motivated by theoretical interest and in part by the importance of this more general case for applications; see~\cite{JMU:Monodromy_preserving_deformation_of_linear_ODEs_with_rational_coefficients1, Jimbo-Miwa:Monodromy_preserving_deformation_of_linear_ODEs_with_rational_coefficients2} and references therein. 
Now, the local solution $Y$ is not meromorphic but can have essential singularities:
\begin{align}\label{eq:non-Fuchsian}
\frac{\partial}{\partial z} Y(z,\bs\lambda; \bs s) = \sum_{i=1}^{n} \sum_{k=0}^{r_i} \frac{A_{i,k}(\bs\lambda; \bs s)}{(z-\lambda_i)^{k+1}} \, Y(z,\bs\lambda; \bs s) ,
\qquad A_{i,k}(\bs\lambda; \bs s) \in \sl_2(\bC) , 
\end{align} 
where $r_i \geq 0$ is the Poincar\'e rank of the singularity $\lambda_i$, and only the leading coefficient matrices $A_{i, r_{i}}$ is assumed to be diagonalizable (see \Cref{sec:evol_of_parameters} and~\cite[Chapter~2]{FIKN:Painleve_transcendents_Riemann-Hilbert_approach}). 
For the isomonodromic flow, there are additional deformations coming from \emph{spectral invariants} of the coefficient matrices of the irregular singularities.
In particular, the local solution $Y(z,\bs\lambda; \bs s)$, the IDEs, and the tau-function $\tau(\bs\lambda; \bs s)$ all depend in a complicated way on 
these spectral parameters, termed (higher) \emph{Birkhoff invariants} 
${\bs s = (s_{i,k})_{\substack{1 \leq i \leq n \\ 1 \leq k \leq r_i}}}$~\cite{BJL:Birkhoff_invariants_and_Stokes_multipliers_for_meromorphic_linear_differential_equations, BHH:Hamiltonian_structure_of_rational_isomonodromic_deformation_systems}.

Interestingly, irregular singularities have not been encountered thus far in association with SLE curves, while in contrast, most integrable systems have such higher order singularities~\cite{Birkhoff:The_generalized_riemann_problem_for_linear_differential_equations_and_the_allied_problems_for_linear_difference_and_q-difference_equations, 
SMJ:Holonomic_quantum_fields3, Jimbo-Miwa:Monodromy_preserving_deformation_of_linear_ODEs_with_rational_coefficients2}.

\subsection{$\SLE_4$ martingale observables, isomonodromy, and integrability\label{subsec:mgle_intro}}

After the introduction of SLEs, a task of fundamental importance has been to provide proofs of convergence and conformal invariance of geometric quantities in critical lattice models. 
The conceptually crucial step is the \emph{identification} of the limiting objects. 
This usually utilizes particular integrability relations that enable the usage of discrete complex analysis techniques.
Important motivation to the scope of the present work are interfaces in the double-dimer model, 
which are believed (with overwhelming evidence) to be described by Schramm's $\SLE_\kappa$ type processes with parameter $\kappa=4$~\cite{Kenyon:Conformal_invariance_of_loops_in_the_double-dimer_model, Dubedat:Double_dimers_conformal_loop_ensembles_and_isomonodromic_deformations, Basok-Chelkak:Tau_functions_a_la_Dubedat_and_probabilities_of_cylindrical_events_for_double_dimers_and_CLE4, Bai-Wan:On_the_crossing_estimates_of_simple_conformal_loops_ensembles}. 
The limit objects can be identified via topological observables: probabilities of connectivities of the chordal curves, and monodromies of the interfaces forming loops in the model.

Towards identification of the possible limit, Kenyon~\&~Wilson computed connectivity probabilities of the chordal curves~\cite{Kenyon-Wilson:Double_dimer_pairings_and_skew_Young_diagrams}; 
and these can be identified with those for $\SLE_4$ curves (see also~\cite{Peltola-Wu:Global_and_local_multiple_SLEs_and_connection_probabilities_for_level_lines_of_GFF,Karrila-Peltola:Boundary_double-dimer_patterns_and_CFT}). 
Moreover, Kenyon proved that the loops (if they converge) should have a conformally invariant scaling limit~\cite{Kenyon:Conformal_invariance_of_loops_in_the_double-dimer_model}.  
The random loop ensemble arising in the scaling limit is termed $\CLE_4$ (\emph{Conformal loop ensemble})~\cite{Sheffield:Exploration_trees_and_CLEs, Sheffield-Werner:CLEs}. 
Characterizing these loops via some geometric observables might seem like a formidable problem --- but \emph{integrability} of the model provided a clue.
Indeed, Kenyon constructed topological correlators associated with the monodromy of the discrete loops and conjectured that they converge to the corresponding $\CLE_4$ observables in the scaling limit, which then should characterize it. 
Dub\'edat later showed in~\cite{Dubedat:Double_dimers_conformal_loop_ensembles_and_isomonodromic_deformations} 
that these observables could be expressed in terms of the isomonodromic tau-function and the associated local solution to the Fuchsian ODE~\eqref{eq:Fuchsian_IDE}. 
The key to the identification of the interface curves is that\footnote{From CFT point of view, this might not come as a huge surprise: the corresponding CFT should have central charge $c = 1$. 
($\SLE_\kappa$ curves should carry central charge $c_\kappa=\frac{1}{2\kappa}{(3\kappa-8)(6-\kappa)}$.)  
Indeed, it has been known for a while in the CFT and integrability communities that conformal blocks for a $c=1$ theory can be expressed in terms of the isomonodromic tau-function of Painlev\'e VI~\cite{GIL:Conformal_field_theory_of_PVI, ILT:Painleve_VI_connection_problem_and_monodromy_of_c_1_conformal_blocks, ILT:Isomonodromic_tau-functions_from_Liouville_conformal_blocks}.
Hence, Dub\'edat's result also gives yet another manifestation of the close relationship of SLE curves with CFT.} 
these quantities give rise to a \emph{martingale observable} for the $\SLE_4$ curve (see also \Cref{sec:mgle}).

Later, Basok~\&~Chelkak proved that these monodromy observables indeed characterize the distribution of the random $\CLE_4$ loops~\cite{Basok-Chelkak:Tau_functions_a_la_Dubedat_and_probabilities_of_cylindrical_events_for_double_dimers_and_CLE4, Bai-Wan:On_the_crossing_estimates_of_simple_conformal_loops_ensembles}
In particular, they related the $\CLE_4$ loop observables to coefficients of the isomonodromic tau-function with respect to the Fock-Goncharov lamination basis on a representation variety 
of the fundamental group of the Riemann sphere with punctures (around which the monodromy is computed) 
with values in $\SL_2(\bC)$~\cite{Fock-Goncharov:Moduli_spaces_of_local_systems_and_higher_Teichmuller_theory}.
These works yield a fundamental geometric interpretation for the loop quantities, where monodromies around punctures serve as the key observables. 

\medskip

A natural conceptual problem now arises: 
\emph{What is the general role of the singularities in the random geometric model?}
The purpose of the present work is to investigate a more general situation where the singularities are allowed to be irregular (non-Fuchsian). 

\medskip

Namely, we will show that $\SLE_4$ processes (and thus also $\CLE_4$) admit martingale observables 
with irregular singularities (\Cref{thm:main_mgle}), and provide a complete description of the Loewner flow of the singularities and the associated Birkhoff spectral invariants (Theorems~\ref{thm:evol_isomonodromic_times}~\&~\ref{thm:Sto-Bir-Gen}).
We also show that the $\SLE_4$ partition functions associated with these irregular martingales solve confluent BPZ equations (\Cref{thm:BPZ_general}), 
matching the known CFT results in the presence of irregular singularities~\cite{Gaiotto-Teschner:Irregular_singularities_in_Liouville_theory_and_Argyres-Douglas_type_gauge_theories, CPZ:Interactions_of_irregular_Gaiotto_states_in_Liouville_theory, HLR:Flat_connections_from_irregular_conformal_blocks} (though these references are in the context of Liouville conformal field theory). 
In particular, our martingale observables can thus be associated 
with confluent chiral degenerate correlation functions with central charge $c=1$.
(This connection, to us, is not a necessary tool.)

Because the SLE flow on the moduli space of Riemann surfaces can be directly related to the isomonodromic flow, we expect that our results pave the way for a geometric approach to SLE curves,
{and to explain the nature of integrability that is deeply linked with SLE theory.}
For instance, tau-functions of certain IDEs can be written as discrete Fourier transforms of conformal blocks on associated Riemann surfaces with prescribed singularities, 
or as Fredholm determinants~\cite{GIL:Conformal_field_theory_of_PVI, 
GIL:How_instanton_combinatorics_solves_Painleve_VI_V_and_IIIs, 
ILT:Isomonodromic_tau-functions_from_Liouville_conformal_blocks, 
BLMST:On_Painleve_gauge_theory_correspondence, 
Gavrylenko-Lisovyy:Fredholm_determinant_and_Nekrasov_sum_representations_of_isomonodromic_tau_functions,  
GIL:Higher_rank_isomonodromic_deformations_and_W-algebras, 
BDMGT:N_2_Gauge_theory_free_fermions_on_the_torus_and_Painleve_VI, 
Desiraju:Fredholm_determinant_representation_of_the_homogeneous_Painleve2, 
Desiraju:Painleve_CFT_correspondence_on_a_torus, 
DMDG:Isomonodromic_tau_functions_on_a_torus_as_Fredholm_determinants_and_charged_partitions}.
These descriptions uncover important asymptotic properties of the tau-functions \cite{ILT:Connection_problem_for_the_Sine-Gordon_Painleve_III_tau_function_and_irregular_conformal_blocks, 
ILP:Monodromy_dependence_and_connection_formulae_for_isomonodromic_tau_functions,
Its-Prokhorov:On_some_Hamiltonian_properties_of_the_isomonodromic_tau_functions, 
DMDG:Monodromy_dependence_and_symplectic_geometry_of_isomonodromic_tau_functions_on_the_torus, 
DMDG:Modular_transformations_of_tau_functions_and_conformal_blocks_on_the_torus}. 
One can also ask what is the precise relation between certain CFT conformal blocks and SLE martingales on various Riemann surfaces (to which we give insight in \Cref{thm:BPZ_general}),
which could reveal new phenomena both in random geometry and in CFT. 

\subsection{Summary of main results}

In this section, we give a very brief summary of the main contributions of this work.

Let $\Driv \colon [0,\infty) \to \bR$ be a continuous real-valued function of time (Loewner driving function). 
Let $(g_t)_{t \geq 0}$ be the associated \emph{Loewner flow}, obtained by solving the ODE
\begin{align}\label{eq:LE}
\begin{cases}
\displaystyle \frac{\partial}{\partial t} \, g_t(z) = \frac{2}{g_t(z) - \Driv_t} , \\
g_0(z) = z ,
\end{cases}
\qquad z \in \hat\bC ,
\end{align}
defined up to the first time $T_z$ when $g_{T_z}(z) = \Driv_{T_z}$. 
For example, if $\gamma \colon [0,\infty) \to \bH$ is a chordal injective curve as in \Cref{fig:SLE}, 
then there exists $\Driv = \Driv(\gamma)$ such that and $\Driv_t = g_t(\gamma(t))$ and the solution $g_t$ to~\eqref{eq:LE} 
is a uniformizing map on the complement of the curve and its reflection across the real axis, normalized in such a way that $g_t(z) \sim z + \frac{2t}{z}$ as $z \to \infty$. 
(In general, for given $(\Driv_t)_{t \geq 0}$ the Loewner ODE~\eqref{eq:LE} has a unique solution up to the blow-up time $T_z$,
but it is not obvious at all what kind sets the evolution describes.)

\subsubsection{Loewner evolution of Birkhoff invariants}

Our first main result concerns the Loewner evolution of Birkhoff invariants. 
For double poles, they transform, in a sense, covariantly under the Loewner flow (see \Cref{thm:evol_isomonodromic_times}). 
In fact, only the second order poles have a description directly in terms of the derivative of the Loewner map --- 
in general, the analytic structure of this evolution is irrespective of the rank of the singularity (\Cref{thm:Sto-Bir-Gen}).
The Birkhoff invariants $s_{i,k}$ are complex numbers related to the eigenvalues of the coefficient matrices in the ODE~\eqref{eq:non-Fuchsian}. 
We write $\hat{g}_t \colon \bC \to \bC$ for the map in the space of the Birkhoff invariants induced by the Loewner flow $(g_t)_{t \geq 0}$: 
the time-evolution of the singularities $\bs\lambda=(\lambda_1,\ldots,\lambda_n)$ induces a time-evolution for the eigenvalues of $A_{i,k}(\bs\lambda;\bs s)$ and thus of the Birkhoff invariants $s_{i,k} \in \bC$, 
which depend in general in a complicated manner on the eigenvalues (see~\cite[Section~3]{BHH:Hamiltonian_structure_of_rational_isomonodromic_deformation_systems} and \Cref{sec:evol_of_parameters}). 

\begin{thm}\label{thm:Sto-Bir-Gen}
Let $(g_t)_{t \geq 0}$ be the Loewner flow with driving function $(\Driv_t)_{t \geq 0}$. 
Consider the Loewner evolution $S_t^{i,k} \coloneqq \hat{g}_t(s_{i,k})$ of the Birkhoff invariant $s_{i,k}$ associated to the ODE~\eqref{eq:non-Fuchsian} corresponding to the pole $\lambda_i$ of order $k+1$.
Assume that its coefficient matrix $A_{i,k}(\bs\lambda;\bs s)$ in~\eqref{eq:non-Fuchsian} is diagonalizable. 
Then, we have
\begin{align}\label{SLEeq-irrgen}
S_t^{i,k} = s_{i,k} \exp \bigg(\! -\int_0^t \frac{2 k \ud u}{(\Lambda_u^i - \Driv_u)^2} \bigg) ,
\qquad \textnormal{i.e.} \quad
\partial_t S_t^{i,k} = -\frac{2 k S_t^{i,k}}{(\Lambda_t^i-\Driv_t)^2} ,
\end{align}
where $\Lambda_u^i = g_u(\lambda_i)$ represents the time-evolution of the singularities $\bs\lambda=(\lambda_1,\ldots,\lambda_n)$. 
\end{thm}

The analytic structure $(\Lambda_t^i-\Driv_t)^{-2}$ of this evolution is independent of the order of the singularity 
for the corresponding Birkhoff invariant $s_{i,k} = S_0^{i,k}$, owing to the purely geometric nature of the latter. 
This fact is crucial for the construction of the $\SLE_4$ martingale.

We prove \Cref{thm:Sto-Bir-Gen} in 
\Cref{subsec:general_poles}: it will be a corollary of the case of double (second order) poles (\Cref{thm:evol_isomonodromic_times} in \Cref{subsec:second_order_poles}).
The crucial ingredient to derive the time-evolution of the Birkhoff invariants is that they can be regarded as the \emph{rate of confluence} of simple poles. 
Indeed, higher order poles can be obtained from simple poles through a confluence procedure~\cite[Section~4]{GMR:Isomonodromic_deformations_confluence_reduction_and_quantisation} --- see \Cref{subsubsec:confluence} for a simple example.

\subsubsection{Irregular SLE martingales}

Using the above evolution, in \Cref{thm:main_mgle} we construct martingale observables for $\SLE_4$ curves at the presence of irregular singularities of rank one (i.e., double poles) on the Riemann sphere.
Interestingly enough, the prefactor of the martingale which governs the conformal covariance 
contains the pre-Schwarzian and Schwarzian of the Loewner map (see Equation~\eqref{eq:Mart}). 
We also provide a possible extension to higher order singularities.

Specifically, we consider a simple $\SLE$ curve $\gamma \colon [0,\infty) \to \bH$ in the upper half-plane $(\bH,0,\infty)$ between the boundary points $0 = \gamma(0)$ and $\infty = \lim_{t \to \infty}\gamma(t)$,
with Loewner driving function $\Driv_t = 2B_t$. 
In the spirit of CFT on surfaces with boundary, we consider the Schottky double: the Riemann sphere $\hat\bC$ obtained 
by gluing $\bH$ with its reflection $\bH^* = \{ z \in \bC \cond \Im(z)<0 \}$ isometrically along their common boundary $\bR \cup \{\infty\}$.  
We consider $n$ distinct punctures $\lambda_1, \lambda_2, \ldots, \lambda_n$ in $\bC$. 
The endpoints $0$ and $\infty$ of the curve are also marked points, playing a different role than the punctures:
the time evolution of the curve started at $0$ determines a flow on the moduli space of $\hat\bC_{\bs\lambda}$, 
and $\infty$ is the target point of the curve, which could be seen as a base point. (See \Cref{fig:SLE}.)

We consider in detail the case where the Poincar\'e rank of each singularity equals $r_i = 1$, and write $(s_{i},-s_{i})$ for the eigenvalues of $A_{i,1}(\bs\lambda;\bs s) \in \sl_2(\bC)$ (the Birkhoff invariants). 
Let $(\alpha_{i},-\alpha_{i})$ be the eigenvalues of $A_{i,0}(\bs\lambda;\bs s) \in \sl_2(\bC)$, playing the role of conformal 
weights\footnote{In the expression~\eqref{eq:Mart}, the functions $F(\bs\Lambda_t; \bs S_t)$ and $\tau(\bs\Lambda_t; \bs S_t)$ are to be understood as the functions of the Loewner evolution of variables $\bs\lambda$, $\bs s$, while $Y(\Driv_t, \bs\Lambda_t; \bs S_t)$ depends on the evolution of $\bs \lambda$, $\bs s$, as well as $z$.}.

\begin{restatable}{thm}{Mgle}
\label{thm:main_mgle}
Let $(g_t)_{t \geq 0}$ be the Loewner flow corresponding to a chordal $\SLE_4$ in $(\bH, 0,\infty)$. 
Then, the following process is a matrix-valued local martingale \textnormal{(}for the natural filtration\footnote{We assume the usual conditions for the probability space and filtration.}\textnormal{):} 
\begin{align}\label{eq:Mart}
M_t \coloneqq 
F(\bs \Lambda_t;  \bs S_t) \, \tau(\bs\Lambda_t; \bs S_t) \, Y(\Driv_t, \bs\Lambda_t; \bs S_t) ,
\end{align}
where $\Driv_t = 2 B_t$ is the Loewner driving function, 
$\bs\Lambda_t = (\Lambda_t^1, \ldots,\Lambda_t^{n})$ are the time-evolutions $\Lambda_t^i = g_t(\lambda_i)$ of the punctures 
and $\bs S_t = (S_t^1, \ldots,S_t^{n})$ the time-evolutions $S_t^i = \hat{g}_t(s_i)$ of the associated Birkhoff invariants,
and 
\begin{align*}
F(\bs \Lambda_t;  \bs S_t) 
\coloneqq \; & \prod_{i=1}^{n} (g_t'(\lambda_i))^{\alpha_i^2} \exp \bigg(\frac{s_i^2}{6} \Schwarzian(g_t)(\lambda_i) + s_i \alpha_i \PreSchwarzian(g_t)(\lambda_i) \bigg) 
\\ 
= \; & \exp \Bigg( \! - 2 \sum_{i=1}^{n} \int_0^t \bigg( 
\frac{\alpha_i^2}{(\Lambda_r^i - \Driv_r)^2}
\; + \; \frac{ (S_r^i)^2 }{(\Lambda_r^i - \Driv_r)^4}
\; - \; \frac{2 \alpha_i S_r^i }{(\Lambda_r^i - \Driv_r)^3} \bigg)  \ud r \Bigg) ,
\end{align*}
where $\PreSchwarzian (g) \coloneqq g''/g'$ is the pre-Schwarzian derivative and $\Schwarzian(g) \coloneqq (g''/g')' - \frac12 (g''/g')^2$ is the Schwarzian derivative of a function $g$. 
\end{restatable}

Interestingly, the computation leading to the local martingale property of the process~\eqref{eq:Mart}  
crucially relies on the identity $\kappa=4$. This is a manifestation of CFTs with central charge $c=1$ having a direct connection to tau-functions.

{
Simple poles can be recovered by setting $s_i=0$. Then, for special values of monodromy data 
one could recover the local martingale constructed by Dub\'edat~\cite{Dubedat:Double_dimers_conformal_loop_ensembles_and_isomonodromic_deformations}.
Moreover, in \Cref{thm:SLE_4(-2)_mart} we also prove an analogous result for an $\SLE_4(\rho)$ process with weight $\rho=-2$, relevant to $\CLE_4$ ensembles and double-dimer interfaces, for example. 
}

\subsubsection{Confluent BPZ differential equations}

Our result only uses stochastic analysis and integrability. 
However, it is possible to arrive at the correct form of the martingale observable from CFT considerations.
Indeed, in the regular case (as also mentioned above), it is now well known that evolution of SLE curves is related to CFT correlation functions of degenerate fields, which in particular satisfy Belavin-Polyakov-Zamolodchikov (BPZ) PDEs. 
We establish an irregular version of this. 

Let $\FunSp(\cX_n)$ denote $C^2$-functions $\PartF$ on the configuration space $\cX_n \subset \bR \times \bC^{2n} \cong \bR^{1+4n}$, 
\begin{align*}
\cX_n := \Big\{ (z, \bs\lambda; \bs s) \in \bR \times \bC^n \times \bC^n \;|\; z \neq \lambda_i, \, \lambda_i \neq \lambda_j \textnormal{ for } 1 \leq i \neq i \leq n \Big\} ,
\end{align*}
holomorphic in the variables $(\bs\lambda; \bs s)$.
Here, the variable $z \in \bR$ represents the tip of the $\SLE_4$ curve. 
The confluent BPZ PDE~\eqref{eq:BPZ} in the below \Cref{thm:BPZ_general} involves Wirtinger derivatives $\frac{\partial}{\partial {\lambda_i}} = \frac{1}{2} \big( \partial_{\Re(\lambda_i)} - \ii \partial_{\Im(\lambda_i)} \big)$ 
and $\frac{\partial}{\partial {s_i}} = \frac{1}{2}  \big( \partial_{\Re(s_i)} - \ii \partial_{\Im(s_i)} \big)$ at the complex variables. 
(Relaxing holomorphicity, the PDE would also include the complex conjugate variables.)

\begin{thm}\label{thm:BPZ_general}
For chordal $\SLE_4$ in $(\bH, 0,\infty)$, the process involving $\PartF \in \FunSp(\cX_n)$, 
\begin{align}\label{eq:Mart_general}
M_t \coloneqq 
F(\bs\Lambda_t; \bs S_t) \, \PartF(\Driv_t, \bs\Lambda_t; \bs S_t) ,
\end{align}
is a local martingale if and only if 
$\PartF \in \FunSp(\cX_n)$ solves the confluent BPZ PDE
\begin{align} \label{eq:BPZ}
\begin{split}
\Bigg[ \frac{\partial^2}{\partial z^2} 
- \sum_{i=1}^{n} 
\; & \Bigg( \frac{1}{z-\lambda_i} \frac{\partial}{\partial {\lambda_i}} + \frac{s_i}{(z-\lambda_i)^2} \frac{\partial}{\partial {s_i}} 
\\
\; & + \frac{\alpha_i^2}{(z-\lambda_i)^2} +  \frac{2 s_i \alpha_i }{(z - \lambda_i)^3} + \frac{s_i^2 }{(z - \lambda_i)^4}  \Bigg) \Bigg]  \PartF(z, \bs\lambda; \bs s) = 0 .
\end{split}
\end{align}
In particular, the function $\tau(\bs\lambda; \bs s) \, \Tr \big( Y(z, \bs\lambda; \bs s) \big)$ solves the confluent BPZ PDE~\eqref{eq:BPZ}. 
\end{thm}

With the holomorphicity assumption, \Cref{thm:BPZ_general} is a direct computation using stochastic calculus:
the confluent BPZ PDE~\eqref{eq:BPZ} holds if and only if $M$ has a vanishing drift.  
In \Cref{thm:BPZ_general}, for notational ease we have stated the simplest case, which also covers the concrete expression $\tau(\bs\lambda; \bs s) \, \Tr \big( Y(z, \bs\lambda; \bs s) \big)$, 
which can be obtained from \Cref{thm:main_mgle}. 
It is quite interesting, however, that PDEs of type~\eqref{eq:BPZ} are also necessary in general for processes of type~\eqref{eq:Mart_general} to be local martingales.
In the case of certain types of regular real singularities, an analogue of this fact appears in~\cite{Dubedat:SLE_and_Virasoro_representations_localization} 
and in later works (e.g.,~\cite{Peltola-Wu:Global_and_local_multiple_SLEs_and_connection_probabilities_for_level_lines_of_GFF, FLPW:Multiple_SLEs_Coulomb_gas_integrals_and_pure_partition_functions}).
To establish a more general result, we use the very recent~\cite{Karrila-Viitasaari:In_prep}, which gives general conditions applicable to our situation. 
We state the result in \Cref{thm:non_BPZ_PDE2} in \Cref{subsec:BPZ}.

\bigskip

{\bf Acknowledgments.}
This project was initiated during and after the Workshop ``Women in Mathematical Physics'' (WOMAP2) in Banff in 2023, and we would like to thank the organizers and the Banff staff for setting up the fruitful workshop. 
We thank Misha Basok for interesting discussions on dimers and isomonodromy, and Marta Mazzocco for useful references and encouragement. 
E.P. thanks Alex Karrila for very useful correspondences about their work in~\cite{Karrila-Viitasaari:In_prep}, and Julien Roussillon for discussions on the confluent BPZ equations.
We also wish to thank Nezhla Aghaei, Ariane Carrance, Elba Garcia Failde, and Alba Grassi for numerous stimulating discussions on related topics. 


H.D.'s work is supported by Marie Sk{\l}odowska-Curie Postdoctoral Fellowship 101203697.

A.K.'s work was supported by Eccellenza grant 194648 of the Swiss National Science Foundation and she was a member of NCCR Swissmap. 

Part of this work has been performed while H.D. and E.P. visited the Institute for Pure and Applied Mathematics (IPAM) in 2024, supported by the National Science Foundation (Grant No. DMS-1925919);
and while the authors visited the Institut Henri Poincar\'e (IHP) in Paris during Fall 2024, which we cordially thank for hospitality and support.
Further work on this project was carried out while the authors visited the Hausdorff Research Institute for Mathematics (HIM) in Bonn during Summer 2025, 
supported by the Deutsche Forschungsgemeinschaft (DFG, German Research Foundation) under Germany's Excellence Strategy EXC-2047/1-390685813.  
E.P.~also acknowledges earlier support from the DFG collaborative research centre ``The mathematics of emerging effects'' CRC-1060/211504053.

This material is part of a project that has received funding from the European Research Council (ERC) under the European Research Council (ERC) under the European Union's Horizon 2020 research and innovation programme (101042460): 
ERC Starting grant ``Interplay of structures in conformal and universal random geometry'' (ISCoURaGe) 
and from the Academy of Finland grant number 340461 ``Conformal invariance in planar random geometry.'' 
While carrying out this project, E.P.~has also been supported by 
the Academy of Finland Centre of Excellence Programme grant number 346315 ``Finnish centre of excellence in Randomness and STructures (FiRST),''
and by the Deutsche Forschungsgemeinschaft (DFG, German Research Foundation) project number 390534769 ``Matter and Light for Quantum Computing (ML4Q).''

%% file: tex-arXiv-v2/2-stoch.tex
In this section, we introduce some basics on the theory of isomonodromic deformations for non-Fuchsian ODEs,
very briefly mention the related integrable structure (the Hamiltonians), and define the tau-function, 
which plays an important role in the present work. 
We first focus on the case of double poles in \Cref{subsec:second_order_poles}, and treat the general case in \Cref{subsec:general_poles}. 
The main results of this section are \Cref{thm:evol_isomonodromic_times} and \Cref{thm:Sto-Bir-Gen}.

The Birkhoff invariants $\bs s = (s_{i,k})$ are spectral parameters arising in the isomonodromic flow from the additional deformations associated to the eigenvalues of the coefficient matrices $A_{i,k}(\bs\lambda;\bs s)$ of the irregular singularities in~(\ref{eq:non-Fuchsian}). 
The higher order poles at $\bs\lambda$, in turn, arise from simple poles through a confluence procedure~\cite[Section~4]{GMR:Isomonodromic_deformations_confluence_reduction_and_quantisation}.
With this insight, one can view each $s_{i,k}$ as the \emph{rate of confluence of simple poles} (see \Cref{subsubsec:confluence}). 

\subsection{Non-Fuchsian equations with double poles\label{subsec:second_order_poles}}

First, we gather some basics in \Cref{subsubsec:double_poles} and discuss the confluence procedure to get higher order poles from simple poles in \Cref{subsubsec:confluence}.
\Cref{subsubsec:second_order_poles_proof} concerns Loewner evolution of the Birkhoff invariants associated to double poles, culminating in \Cref{thm:evol_isomonodromic_times}.

\subsubsection{Birkhoff spectral invariants and tau-function for double poles}
\label{subsubsec:double_poles}

Let $\bs\lambda=(\lambda_1,\ldots,\lambda_n)$ be $n$ punctures on the Riemann sphere $\hat\bC \coloneqq \bC \cup \{\infty\}$.
Let us consider the case of double poles and the non-Fuchsian equation~\eqref{eq:non-Fuchsian} in the form
\begin{align}\label{eq:non-Fuchsian-DP}
\begin{split}
\frac{\partial}{\partial z} Y(z) = \; & A(z) \, Y(z) ,
\qquad z\in \hat\bC_{\bs\lambda} \coloneqq \hat\bC \setminus \{\lambda_1,\ldots,\lambda_n\} ,
\\
A(z) \coloneqq \; & \sum_{i=1}^{n} \bigg( \frac{A_{i,0}}{z-\lambda_i} + \frac{A_{i,1}}{(z-\lambda_i)^{2}} \bigg) , \qquad A_{i,k} \in \sl_2(\bC) ,
\end{split}
\end{align}
where the (in our case, traceless) matrix-valued function $A(z)$ is dubbed the \emph{Lax matrix}.
The local solution $Y(z)$ is a holomorphic $(2\times 2)$-matrix valued function. 
For each $i$, the eigenvalue $s_{i}$ of the leading coefficient matrix $A_{i,1}$ with non-negative real part 
is called the \emph{Birkhoff (spectral) invariant} associated to the puncture $\lambda_i$.  
For ease, we assume throughout that the matrices $A_{i,1} \in \sl_2(\bC)$ are non-zero and diagonalizable, such that 
\begin{align}\label{def:diagLij}
A_{i,1} = - G_{i} \, D_{i,1} \, G_{i}^{-1} , 
\qquad \textnormal{where} \qquad 
D_{i,1} \coloneqq 
\left(
\begin{matrix}
s_i & 0 \\
0 & -s_i
\end{matrix}
\right)  , \qquad \Re(s_{i}) \geq 0 ,
\end{align}
and $G_{i} \in \GL_2(\bC)$. Without loss of generality, locally near $z=\lambda_i$, we can set $G_{i}= \mathbf{1}$. A notable example with linear ODEs with the above structure is the Painlev\'e V equation.

The (formal) solution $Y$ to the ODE~\eqref{eq:non-Fuchsian-DP} at the vicinity of the singularity $\lambda_i$ reads
\begin{align}\label{eq:formY}
Y(z) = G_{i} \, \bigg( \mathbf{1} + \sum_{j=1}^{\infty} N_j (z-\lambda_i)^j \bigg) 
\, (z-\lambda_i)^{D_{i,0}} 
\, \exp\bigg(\frac{D_{i,1}}{z-\lambda_i} \bigg),
\end{align}
and the matrices $N_j$ can be computed recursively by plugging~\eqref{eq:formY} into~\eqref{eq:non-Fuchsian-DP} (cf.~\cite{FIKN:Painleve_transcendents_Riemann-Hilbert_approach}),
\begin{align}\label{eq:def_Dio}
D_{i,0} \coloneqq 
\left(
\begin{matrix}
\alpha_i & 0 \\
0 & -\alpha_i
\end{matrix}
\right) , \qquad \Re(\alpha_{i}) \geq 0 ,
\end{align}
where the eigenvalues $(\alpha_{i},-\alpha_{i})$ of the coefficient matrices $A_{i,0}$ of the Fuchsian term are called \emph{exponents of formal monodromy} of $Y$. 
The matrix $D_{i,0}$ is related to $A_{i,0}$ as 
\begin{align}\label{eq:Ai0-GiDio}
A_{i,0} = G_i \big(D_{i,0}  + [D_{i,1},N_1]\big)G_i^{-1}.
\end{align}
Crucially the exponents of formal monodromy are constant in isomonodromic deformations.
In the martingale observables, they will play the role of conformal weights of primary fields. 

\begin{remark}\label{rem:alpha_const}
By the definition of isomonodromy, the eigenvalues $(\alpha_{i},-\alpha_{i})$ of the coefficient matrices $A_{i,0} \in \sl_2(\bC)$ in the IDEs~\textnormal{\eqref{eq:non-Fuchsian-DP}} are constant:
\begin{align*}
\tfrac{\partial}{\partial \lambda_j} \alpha_i(\bs\lambda; \bs s) = \tfrac{\partial}{\partial s_j} \alpha_i(\bs\lambda; \bs s) = 0 , \qquad \textnormal{for all } 1 \leq i,j \leq n .
\end{align*}
\end{remark}

From the local behavior of the solution $Y = Y(z,\bs\lambda; \bs s)$ near the singularities~\eqref{eq:formY}, 
one can derive the following linear equation with respect to the positions of the poles $\lambda_{i}$: 
\begin{align}\label{eq:LaxUi}
\tfrac{\partial}{\partial \lambda_{i}} Y(z,\bs\lambda; \bs s) = U_{i}(z,\bs\lambda;\bs s) \, Y(z,\bs\lambda; \bs s) , \qquad 
U_{i}(z,\bs\lambda;\bs s) = - \frac{A_{i,1}(\bs\lambda;\bs s)}{(z-\lambda_i)^2} - \frac{A_{i,0}(\bs\lambda;\bs s)}{z-\lambda_i} 
\end{align}
(see~\cite{FIKN:Painleve_transcendents_Riemann-Hilbert_approach} or~\cite[Proposition~3.1 and Section~3.2]{BHH:Hamiltonian_structure_of_rational_isomonodromic_deformation_systems} for a summary). 
Due to the presence of the double pole, in addition to the positions of the poles $\lambda_{i}$, the eigenvalues $s_i$ of 
the leading matrix coefficients $A_{i,1}$ act as \emph{isomonodromic deformation parameters}:
\begin{align}\label{eq:LaxVi}
\tfrac{\partial}{\partial s_{i}} Y(z,\bs\lambda; \bs s) = V_{i}(z,\bs\lambda; \bs s) \, Y(z,\bs\lambda; \bs s) , \qquad 
s_{i} V_{i}(z,\bs\lambda; \bs s) 
= - \frac{A_{i,1}(\bs\lambda;\bs s)}{(z-\lambda_i)} .
\end{align}
In fact, as discussed in~\cite{BHH:Hamiltonian_structure_of_rational_isomonodromic_deformation_systems}, 
the parameters $s_{i}$ can be written alternatively in the form 
\begin{align}\label{eq:isomonodromic_times}
s_{i} = - \underset{z=\lambda_i}{\res}\, (z-\lambda_i) \, \sigma(z,\bs\lambda; \bs s) ,
\end{align}
where $\sigma = \sigma(z,\bs\lambda; \bs s)$ is defined through the characteristic polynomial of $A = A(z,\bs\lambda;\bs s)$, 
\begin{align*}
\det \big( A(z,\bs\lambda;\bs s) - \sigma(z,\bs\lambda; \bs s) \mathbf{1} \big) = 0 .
\end{align*}
Furthermore, one can obtain identities between the matrices $A$, $U_i$, and $V_i$ by studying the compatibility conditions between the equations~(\ref{eq:non-Fuchsian-DP},~\ref{eq:LaxUi},~\ref{eq:LaxVi}); for instance, because the derivatives of $Y$ commute, we have
\begin{align}\label{eq:comp1}
\begin{split}
\tfrac{\partial}{\partial z} \tfrac{\partial}{\partial \lambda_{i}} Y(z,\bs\lambda; \bs s)
= \; &  \tfrac{\partial}{\partial \lambda_{i}} \tfrac{\partial}{\partial z} Y(z,\bs\lambda; \bs s)
\\
\qquad \Longrightarrow \qquad \; & 
 \tfrac{\partial}{\partial \lambda_{i}} A(z,\bs\lambda;\bs s)-\tfrac{\partial}{\partial z} U_i(z,\bs\lambda;\bs s) 
+ \big[A(z,\bs\lambda;\bs s), U_i(z,\bs\lambda;\bs s) \big] = 0 .
\end{split}
\end{align}
Similarly, the compatibility of the equations~\eqref{eq:non-Fuchsian-DP}, and~\eqref{eq:LaxVi} gives 
\begin{align}\label{eq:comp2}
 \tfrac{\partial}{\partial \lambda_{i}} A(z,\bs\lambda;\bs s)- \tfrac{\partial}{\partial z} V_i(z,\bs\lambda;\bs s) 
+ \big[A(z,\bs\lambda;\bs s), V_i(z,\bs\lambda;\bs s) \big] = 0 .
\end{align}
Equations~\eqref{eq:comp1} can also be viewed as a flatness condition for connections on the moduli space of the punctured Riemann surface $\hat\bC_{\bs\lambda}$ 
defined through the ODEs~(\ref{eq:non-Fuchsian-DP},~\ref{eq:LaxUi},~\ref{eq:LaxVi}). 

\medskip

We will make use of the following set of equations for the coefficient matrices in~\eqref{eq:non-Fuchsian-DP},
which follow immediately from the compatibility equations~(\ref{eq:comp1},~\ref{eq:comp2}).

\begin{lem}\label{lem:compatibility_conditions}
For $1 \leq i \neq j \leq  n$, we have
\begin{align*}
\tfrac{\partial}{\partial\lambda_i} A_{j,0} 
\; = \;\; & - \frac{2 \, [A_{i,1}, A_{j,1}]}{(\lambda_i - \lambda_j)^3} \; - \; \frac{[A_{i,1}, A_{j,0}] - [A_{i,0}, A_{j,1}]}{(\lambda_i - \lambda_j)^2} \; + \; \frac{[A_{i,0}, A_{j,0}]}{\lambda_i - \lambda_j} ; \\
\tfrac{\partial}{\partial\lambda_i} A_{i,0} 
\; = \;\; &  \sum_{j\neq i}\bigg(\frac{2 \, [A_{i,1}, A_{j,1}]}{(\lambda_i - \lambda_j)^3} \; + \; \frac{[A_{i,1}, A_{j,0}] - [A_{i,0}, A_{j,1}]}{(\lambda_i - \lambda_j)^2} \; - \; \frac{[A_{i,0}, A_{j,0}]}{\lambda_i - \lambda_j} \bigg) ; \\
\tfrac{\partial}{\partial\lambda_i} A_{j,1} 
\; = \;\; & - \frac{[A_{i,1}, A_{j,1}]}{(\lambda_i - \lambda_j)^2} \; + \; \frac{[A_{i,0}, A_{j,1}]}{\lambda_i - \lambda_j} ; \\
\tfrac{\partial}{\partial\lambda_i} A_{i,1} 
\; = \;\; & -\sum_{j\neq i} \bigg( \frac{[A_{i,1}, A_{j,1}]}{(\lambda_i - \lambda_j)^2} \; + \; \frac{[A_{i,1}, A_{j,0}]}{\lambda_i - \lambda_j}\bigg) ; \\
s_i \, \tfrac{\partial}{\partial s_i} A_{j,0} 
\; = \;\; & -\frac{\big[A_{j,1} ,  A_{i,1}\big]}{(\lambda_i - \lambda_j)^2} 
\; - \; \frac{\big[A_{j,0} ,  A_{i,1}\big]}{\lambda_i - \lambda_j} ; \\
s_i \,  \tfrac{\partial}{\partial s_i} A_{j,1} 
\; = \;\; & -\frac{\big[A_{j,1}, A_{i,1}\big]}{\lambda_i - \lambda_j} ; \\
s_i \, \tfrac{\partial}{\partial s_i}A_{i,1} \; =\;\; & A_{i,1}+ [A_{i,0}, A_{i,1} ]\\
s_i \, \tfrac{\partial}{\partial s_i} A_{i,0} \; = \;\; & -\frac{[ A_{i,1}, A_{j,1}]}{(\lambda_{i} - \lambda_{j})^2}-\frac{[ A_{i,1}, A_{j,0}]}{(\lambda_{i} - \lambda_{j})}.
\end{align*}
\end{lem}

The construction of the $\SLE_4$ martingale depends on the solution of the Fuchsian equation~\eqref{eq:non-Fuchsian-DP} and the associated tau-function. 
The latter is inherently related to a Poisson structure for the integrable system arising from the isomonodromy, 
encoded into Hamiltonians\footnote{The factor ``$2$'' in the Hamiltonian $H_{s_i}$ is due to the fact that the eigenvalues of $A_{i,1}$ are $\pm s_i$.}  
$H_{\lambda_i}$ and $H_{s_i}$ corresponding to each of the deformation parameters $\lambda_i$ and $s_i$: 
\begin{align}
\label{def:HamLambda}
H_{\lambda_i} 
\coloneqq \; & \frac{1}{2}  \, \underset{z=\lambda_i}{\res}\, \Tr \big( (A(z,\bs\lambda;\bs s))^2 \big) , 
\\
\label{def:HamS}
H_{s_i} \coloneqq \; & -2 \, \underset{z=\lambda_i}{\res}\, \frac{\sigma(z,\bs\lambda; \bs s)}{z-\lambda_i} .
\end{align}
The tau-function associated the system~(\ref{eq:non-Fuchsian-DP},~\ref{eq:LaxUi},~\ref{eq:LaxVi}) of equations is defined through its logarithmic derivative as 
\begin{align}\label{def:tau}
\big( \tfrac{\partial}{\partial {\bullet}} \log \tau(\bs\lambda; \bs s) \big) \ud \bullet \coloneqq \; & H_{\bullet}, \qquad \bullet = \{ s_{i}, \lambda_i\} .
\end{align}
Observe that the tau-function is defined only up to an overall constant, which is uniquely determined by the choice of monodromy data. 
We refer to~\cite{Balogh-Harnad:Tau_functions_and_their_applications} for a detailed discussion about the tau-function and its geometrical meaning. 

\begin{lem}\label{lem:Hamiltonians}
For the Lax matrix $A$ in~\eqref{eq:non-Fuchsian-DP} with double poles, we have
\begin{align}
\label{eq:TrL2}
\Tr \big( (A(z, \bs\lambda; \bs s))^2 \big) \; & = \sum_{i=1}^{n} \sum_{k=1}^{4} \frac{\ell_{i,k}(\bs\lambda; \bs s)}{(z-\lambda_i)^{k}} ,
\qquad \textnormal{where}
\\
\label{eq:ellval}
 \; & 
\begin{cases}
\ell_{i,4} = \ell_{i,4}(\bs s) = 2 \, s_i^2 , \\
\ell_{i,3} = \ell_{i,3}(\bs\lambda; \bs s) = 2 \, \Tr \big(A_{i,0}(\bs\lambda;\bs s) \, A_{i,1}(\bs\lambda;\bs s) \big) 
= 4 \, s_{i} \alpha_{i} .
\end{cases}
\end{align}
The Hamiltonians defined in~\textnormal{(\ref{def:HamLambda},~\ref{def:HamS})} can be written in the form
\begin{align}
\label{eqell:Hlamb}
H_{\lambda_i} = \; & \tfrac{1}{2} \, \ell_{i,1} ,
\\
\label{eqell:Hsgen}
H_{s_i} = \; & \bigg(\frac{\ell_{i,2}}{\ell_{i,4}} - \frac{\ell_{i,3}^2}{4 \ell_{i,4}^2}  \bigg) \bigg( \frac{\ell_{i,4}}{2} \bigg)^{1/2}
{= \frac{\ell_{i,2}}{2 s_i} - \frac{\ell_{i,3}^2}{16 s_i^3}.}
\end{align}
\end{lem}

\begin{proof}
The leading coefficient in~\eqref{eq:ellval} for $A$ is obtained as a consequence of the diagonalizability of $A_{i,1}$ in~\eqref{def:diagLij}.
To compute the subleading one, note that 
\begin{align*}
\ell_{i,3} 
=  2 \, \Tr \big( A_{i,0} \, A_{i,1} \big) 
\overset{\textnormal{\eqref{def:diagLij}}}{\underset{\textnormal{\eqref{eq:Ai0-GiDio}}}{=}}
\; & - 2 \, \Tr \Big( \big(D_{i,0} + [D_{i,1}, N_1 ] \big)D_{i,1} \Big) 
\\
\overset{\hphantom{\textnormal{\eqref{def:diagLij}}}}{\underset{\hphantom{\textnormal{\eqref{eq:Ai0-GiDio}}}}{=}} \; & - 2 \, \Tr  \big(D_{i,0}D_{i,1} \big)
\overset{\textnormal{\eqref{def:diagLij}}}{\underset{\textnormal{\eqref{eq:def_Dio}}}{=}}
4 s_i \alpha_i ,
\end{align*}
also using the fact that, because $D_{i,0}$ is diagonal, the diagonal entries of the matrix $[D_{i,1}, N_1 ]$ equal zero. Let us then elaborate on the structure of the Hamiltonians. 
From~(\ref{def:HamLambda},~\ref{eq:TrL2}),
\begin{align*}
H_{\lambda_i} 
= \; & \frac{1}{2}  \, \underset{z=\lambda_i}{\res}\, \sum_{j=1}^{n} \sum_{k=1}^{4} \frac{\ell_{j,k}}{(z-\lambda_j)^{k}}
= \tfrac{1}{2} \, \ell_{i,1} .
\end{align*}
Since $A = A(z,\bs\lambda;\bs s) \in \sl_2(\bC)$, 
we have the Cayley-Hamilton identity
\begin{align}\label{id:CH}
A^2 = \Tr(A) \, A - \det (A) \, \mathbf{1} 
= - \det (A) \, \mathbf{1} 
= \tfrac{1}{2} \Tr(A^2) \, \mathbf{1} .
\end{align}
Therefore, using the fact that $A(z, \bs{\lambda},\bs s)$ is traceless, we get the relation
\begin{align*}
(\sigma(z,\bs\lambda; \bs s))^2 
= \Big( \! - \det \big(A(z,\bs\lambda;\bs s)\big) \Big)
= \; & \tfrac{1}{2} \Tr\big( (A(z,\bs\lambda;\bs s))^2\big),
\end{align*}
and the identity~\eqref{eq:isomonodromic_times} further dictates that
\begin{align}\label{id:sigma-TrA2}
\sigma(z,\bs\lambda; \bs s)  
= \; & - \Big(\tfrac{1}{2} \Tr\big( (A(z,\bs\lambda;\bs s))^2\big) \Big)^{1/2}.
\end{align}
Using~(\ref{def:HamS},~\ref{eq:TrL2}), we now compute $H_{s_i}$, which has the form~\eqref{def:HamS}, by analysing the expansion of the term below for $z\to \lambda_i$:
\begin{align} 
\frac{\big(\Tr\big( (A(z,\bs\lambda;\bs s))^2\big) \big)^{1/2}}{z-\lambda_i} =\; &\bigg(  \sum_{j=1}^{n} \sum_{k=1}^{4} \frac{\ell_{j,k}(\bs \lambda;\bs s)}{(z-\lambda_j)^{k}} \bigg)^{1/2} \frac{1}{z-\lambda_i} \nonumber \\
= \; & \frac{\ell_{i,4}^{1/2}}{(z-\lambda_i)^3} \Big(1+\frac{\ell_{i,3}}{\ell_{i,4}}(z-\lambda_i)+ \frac{\ell_{i,2}}{\ell_{i,4}}(z-\lambda_i)^2 + O((z-\lambda_i)^3) \Big)^{1/2} \nonumber\\
=\; & \frac{\ell_{i,4}^{1/2}}{(z-\lambda_i)^3} + \frac{\ell_{i,4}^{-1/2} \ell_{i,3}}{2 (z-\lambda_i)^2} + \frac{\ell_{i,4}^{1/2}}{2 (z-\lambda_i)}\bigg(\frac{\ell_{i,2}}{\ell_{i,4}} - \frac{\ell_{i,3}^2}{4 \ell_{i,4}^2}  \bigg)+O(1). \label{asymp:TrA2}
\end{align}
Therefore, we conclude that
\begin{align*}
H_{s_{i}} 
= -2 \, \underset{z=\lambda_i}{\res} \frac{\sigma(z,{\bf \lambda};s)}{z- \lambda_i} 
\overset{\textnormal{\eqref{id:sigma-TrA2}}}{=}
\; & \sqrt{2} \, \underset{z=\lambda_i}{\res} \frac{\big(\Tr\big( (A(z,\bs\lambda;\bs s))^2\big) \big)^{1/2}}{z- \lambda_i} 
\\
\overset{\textnormal{\eqref{asymp:TrA2}}}{=} \; &  \bigg(\frac{\ell_{i,2}}{\ell_{i,4}} - \frac{\ell_{i,3}^2}{4 \ell_{i,4}^2}  \bigg) \bigg( \frac{\ell_{i,4}}{2} \bigg)^{1/2},
\end{align*}
which finishes the proof. 
\end{proof}

\subsubsection{Confluence of simple poles}
\label{subsubsec:confluence}

Let us now describe the confluence procedure for poles of order two~\cite[Section~4]{GMR:Isomonodromic_deformations_confluence_reduction_and_quantisation}. 
Consider the Fuchsian equation~\eqref{eq:Fuchsian} with $n$ simple poles ${\bs \lambda} = (\lambda_1,\ldots, \lambda_n)$,
\begin{align*}
\frac{\partial}{\partial z} Y(z, {\bs \lambda}) 
= \bigg( \sum_{i=1}^{n}\frac{A_i({\bs \lambda})}{z-\lambda_i}\bigg) \, Y(z,{\bs \lambda}).
\end{align*}
Fix $\lambda_1 \in \bC$ and suppose that $\lambda_2 = \lambda_1 + \epsilon \, s$ for some small $\epsilon >0$ and fixed $s \in \bC \setminus \{0\}$. 
Then, a simple computation using the geometric series yields\footnote{With a slight abuse of notation, we denote $A_i({\bs \lambda}, \epsilon \,s)\equiv A_i(\lambda_1, \lambda_1 + \epsilon\, s, \lambda_3,\ldots,\lambda_n)$.}
\begin{align*}
\frac{A_1({\bs\lambda})}{z-\lambda_1} \; + \; \frac{ A_2({\bs \lambda})}{z-\lambda_2}
\; = \; \frac{A_1({\bs\lambda},\epsilon \, s) +  A_2({\bs \lambda}, \epsilon \, s)}{z-\lambda_1} \; + \; \frac{\epsilon \, s \, A_2({\bs \lambda}, \epsilon \, s)}{(z - \lambda_1)^2} \; + \; O(\epsilon^2) .
\end{align*}
Taking $\epsilon \to 0$ and defining (assuming the limits exist and are independent of $s$)
\begin{align} \label{eqn:confluence_limits}
\begin{split}
A_{1,0}({\bs \lambda},s) \coloneqq \; & \lim_{\epsilon \to 0} \Big( A_1({\bs \lambda},\epsilon \, s) +  A_2({\bs \lambda}, \epsilon \, s) \Big) \; \in \; \sl_2(\bC) , \\
A_{1,1}({\bs\lambda},s) \coloneqq \; & \lim_{\epsilon \to 0} \epsilon \, s\, A_2({\bs \lambda}, \epsilon \, s) \; \in \; \sl_2(\bC) ,\\
A_{i,0}({\bs \lambda},s) \coloneqq \; & \lim_{\epsilon \to 0}  A_i({\bs \lambda},\epsilon \, s) +  A_2({\bs \lambda}, \epsilon \, s) \; \in \; \sl_2(\bC) , \; i=3,\ldots,n
\end{split}
\end{align}
yields the non-Fuchsian equation
\begin{align*}
\frac{\partial}{\partial z} Y(z,\lambda; s) 
= \bigg( \frac{A_{1,0}({\bs\lambda},s)}{z-\lambda_1} \; + \; \frac{s \, A_{1,1}({\bs\lambda},s)}{(z-\lambda_1)^{2}} \;+ \; \sum_{i=3}^n \frac{A_{i,0}({\bs \lambda},s)}{(z-\lambda_i)} \bigg)
\, Y(z,\lambda; s).
\end{align*}
Assume also that the matrix $A_{1,1}(\lambda)$ is diagonalizable as in~\eqref{def:diagLij}, 
such that its eigenvalues equal $\pm s$ with $\Re(s)>0$. Then, $s$ is indeed the lowest order Birkhoff invariant.  

\begin{remark} \label{rem:legitimate}
The assumptions~\eqref{eqn:confluence_limits} are perfectly legitimate. 
Indeed, for given desired matrices $A_{1,0}, A_{1,1} \in \sl_2(\bC)$ and parameter $s \in \bC \setminus \{0\}$, the linear equations
\begin{align*}
\left(
\begin{matrix}
1 & 1 \\
0 & \varepsilon
\end{matrix}
\right)
\left(
\begin{matrix}
A_1 \\
A_2
\end{matrix}
\right)
= 
\left(
\begin{matrix}
A_{1,0} \\ 
A_{1,1}
\end{matrix}
\right) , \qquad \varepsilon > 0 ,
\end{align*}
are clearly non-degenerate, so it is possible to find sequences $(A_1^{(\varepsilon)})_{\varepsilon>0}$ and $( A_2^{(\varepsilon)})_{\varepsilon>0}$ such that
$A_{1,0} = A_1^{(\varepsilon)} + A_2^{(\varepsilon)}$ and $ A_{1,1} = \varepsilon \, A_2^{(\varepsilon)}$.
Taking $\varepsilon = \epsilon \, s$ yields matrices satisfying~\eqref{eqn:confluence_limits}. 
\end{remark}

\subsubsection{Loewner evolution of Birkhoff invariants for double poles}
\label{subsubsec:second_order_poles_proof}

To begin, let us consider the case where the Poincar\'e rank of each singularity equals $r_i = 1$, and write $s_{i,1} =: s_i$ for the Birkhoff invariants, 
which in this case are simply given by the eigenvalues $(s_{i},-s_{i})$ of the matrices $A_{i,1}(\bs\lambda;\bs s) \in \sl_2(\bC)$. 
These Birkhoff invariants of double poles transform covariantly under the Loewner flow in the following sense.

\begin{thm}\label{thm:evol_isomonodromic_times}
Let $(g_t)_{t \geq 0}$ be the Loewner flow with driving function $(\Driv_t)_{t \geq 0}$. 
Consider the Loewner evolution $\bs S_t = (S_t^1, \ldots,S_t^n)$ of the Birkhoff invariants $\bs s = (s_1,\ldots,s_n)$ associated to the ODE~\eqref{eq:non-Fuchsian} with double poles \textnormal{(}$r_i=1$ for all $i$\textnormal{)}.
Assume that its coefficient matrix $A_{i,1}(\bs\lambda;\bs s)$ in~\eqref{eq:non-Fuchsian} is diagonalizable. 
Then, we have
\begin{align}\label{eqn:evol_isomonodromic_times}
S_t^{i} \coloneqq \hat{g}_t(s_i) = g_t'(\lambda_i) s_i = s_i \exp \bigg(\! -\int_0^t \frac{2 \ud u}{(\Lambda_u^i - \Driv_u)^2} \bigg) ,
\end{align}
where $\Lambda_u^i = g_u(\lambda_i)$ represents the time-evolution of the singularities $\bs\lambda=(\lambda_1,\ldots,\lambda_n)$. 
Thus, 
\begin{align*}
\partial_t S_t^{i} = - \frac{2 S_t^{i}}{(\Lambda_t^i - \Driv_t)^2} .
\end{align*}
\end{thm}

\begin{proof}
Let $\Lambda_t^i \coloneqq g_t(\lambda_i)$ be the time-evolution of the puncture $\lambda_i$,
and let $\tilde \Lambda_t^i \coloneqq g_t(\tilde \lambda_i)$ be the time-evolution of the puncture $\tilde \lambda_i = \lambda_i + \epsilon \, s_i$. 
In the confluence limit, we have 
\begin{align*}
S_t^{i} = \lim_{\epsilon \to 0} \frac{\tilde \Lambda_t^i - \Lambda_t^i}{\epsilon} 
= \lim_{\epsilon \to 0} \frac{g_t(\lambda_i + \epsilon \, s_i) - g_t(\lambda_i)}{\epsilon}
= s_{i} \, \lim_{\epsilon \to 0} \frac{g_t(\lambda_i + \epsilon) - g_t(\lambda_i)}{\epsilon}
= g_t'(\lambda_i) s_i .
\end{align*}
This shows the first equation in~\eqref{eqn:evol_isomonodromic_times}. 
The second follows from the Loewner ODE~\eqref{eq:LE}.
\end{proof}

\subsection{Non-Fuchsian equations with general pole structure\label{subsec:general_poles}}

In what follows, we summarize the basic setup in the general case and then finish with the proof of \Cref{thm:Sto-Bir-Gen}, 
which is a consequence of \Cref{thm:evol_isomonodromic_times} and the confluence procedure. 

\subsubsection{Birkhoff invariants and tau-function}

More generally, consider the non-Fuchsian equation~\eqref{eq:non-Fuchsian},
\begin{align}\label{eq:non-Fuchsian-A}
\begin{split}
\frac{\partial}{\partial z} Y(z,\bs\lambda; \bs s) = \; & A(z,\bs\lambda;\bs s) \, Y(z,\bs\lambda; \bs s) ,
\qquad z\in \hat\bC_{\bs\lambda} \coloneqq \hat\bC \setminus \{\lambda_1,\ldots,\lambda_n\} ,
\\
A(z,\bs\lambda;\bs s) \coloneqq \; & \sum_{i=1}^{n} \sum_{k=0}^{r_i} \frac{A_{i,k}(\bs\lambda;\bs s)}{(z-\lambda_i)^{k+1}} , \qquad A_{i,k}(\bs\lambda;\bs s) \; \in \; \sl_2(\bC) ,
\end{split}
\end{align}
with Lax matrix $A(z,\bs\lambda;\bs s)$, where we again assume that the leading matrix coefficients $A_{i,r_i} \in \sl_2(\bC)$ are non-zero and diagonalizable, so there exist $G_{i,r_i} \in \GL_2(\bC)$ such that
\begin{align}\label{def:diagLij_gen}
A_{i,r_i} = - G_{i,r_i} \, D_{i,r_i} \, G_{i,r_i}^{-1} , 
\qquad \textnormal{where} \qquad 
D_{i,r_i} \coloneqq 
\left(
\begin{matrix}
s_{i,r_i} & 0 \\
0 & -s_{i,r_i}
\end{matrix}
\right)  , \qquad \Re(s_{i,r_i}) \geq 0 .
\end{align}
Analogously to~\eqref{eq:isomonodromic_times}, we define the spectral invariants as 
\begin{align*}
s_{i,k} = - \underset{z=\lambda_i}{\res}\, (z-\lambda_i)^{k} \, \sigma(z,\bs\lambda; \bs s) ,
\end{align*}
where $\sigma$ is defined through the characteristic polynomial via 
$\det \big( A(z,\bs\lambda;\bs s) - \sigma(z,\bs\lambda; \bs s) \mathbf{1} \big) = 0$,
and the choice $\Re(s_{i,r_i})>0$ determines it uniquely. 
We define the Hamiltonians 
\begin{align*}
H_{\lambda_i} 
\coloneqq \frac{1}{2}  \, \underset{z=\lambda_i}{\res}\, \Tr \big( (A(z,\bs\lambda;\bs s))^2 \big) 
\qquad \textnormal{and} \qquad
H_{s_{i,k}}
\coloneqq -2 \, \underset{z=\lambda_i}{\res}\, \frac{\sigma(z,\bs\lambda; \bs s)}{k (z-\lambda_i)^k} .
\end{align*} 
Then, as before, the tau-function $\tau(\bs\lambda; \bs s)$ is defined as in~\eqref{def:tau}. 
The formal solution $Y$ to the ODE~\eqref{eq:non-Fuchsian-A} at the vicinity of the singularity $\lambda_i$ can be written in the form~\eqref{eq:formY}:
\begin{align}\label{eq:formY_gen}
Y(z) = G_{i,r_i} \, \bigg( \mathbf{1} + \sum_{j=1}^{\infty} N_j (z-\lambda_i)^j \bigg) 
\, (z-\lambda_i)^{D_{i,0}} 
\, \exp\bigg(\sum_{k=1}^{r_i} \frac{D_{i,k}}{k \, (z-\lambda_i)^{k}} \bigg),
\end{align}
where the diagonal matrix $D_{i,0}$ with eigenvalues $\pm \alpha_i$ parametrizes the formal monodromy obtained from the eigenvalues $\pm \alpha_i$ of the coefficient matrices $A_{i,0}$ of the Fuchsian term, 
and $D_{i,k}$ are diagonal matrices with eigenvalues $\pm s_{i,k}$, and $N_j$ are some matrices that can be computed through a recursive relation~\cite{FIKN:Painleve_transcendents_Riemann-Hilbert_approach}.
The explicit relation between the diagonal matrices $D_{i,k}$ and the coefficient matrices $A_{i,k}$ can be obtained by substituting~\eqref{eq:formY_gen} back into the ODE~\eqref{eq:non-Fuchsian-A} and matching powers of $(z-\lambda_i)$.
We read out the first two relations:
\begin{align*}
A_{i, r_i} = - G_{i,r_i} \, D_{i,r_i} \, G_{i,r_i}^{-1} 
\qquad \textnormal{and} \qquad
A_{i, r_i-1} = - G_{i,r_i} \, \Big( D_{i,r_i-1}- G_{i,r_i} \big[ N_1 , D_{i,r_i} \big] \Big) \, G_{i,r_i}^{-1} .
\end{align*}

From~\eqref{eq:formY_gen}, we can compute the higher analogue of the deformation equations~(\ref{eq:LaxUi},~\ref{eq:LaxVi}):
\begin{align}
\label{eq:LaxUigen}
&\tfrac{\partial}{\partial \lambda_{i}} Y(z,\bs\lambda; \bs s) 
=  \;  U_i(z,\bs\lambda;\bs s) \, Y(z,\bs\lambda; \bs s) , \qquad 
U_i(z,\bs\lambda;\bs s) = -  \sum_{k=0}^{r_i} \frac{A_{i,k}(\bs\lambda;\bs s)}{(z-\lambda_i)^{k+1}} , 
\\
\label{eq:LaxVigen}
&  \tfrac{\partial}{\partial s_{i,k}} Y(z,\bs\lambda; \bs s)
= V_{i,k}(z,\bs\lambda;\bs s) Y(z,\bs\lambda; \bs s), \,\,\,\,\,\,  \sum_{k=1}^{r_i} k \, s_{i,k}V_{i,k}= \,- \sum_{k=1}^{r_i} \frac{A_{i,k}(\bs\lambda; \bs s)}{ (z-\lambda_i)^{k}}.
\end{align}
The identity \eqref{eq:LaxVigen} will be key to construct $\SLE_4$ martingales in the next section.

\subsubsection{Loewner evolution: the general case}

For a Fuchsian equation~\eqref{eq:Fuchsian} with $k+1$ simple poles $\bs \lambda = (\lambda_1, \ldots, \lambda_{k+1})$, 
\begin{align*}
\frac{\partial}{\partial z} Y(z, \bs \lambda) = \sum_{i=1}^{k+1} \frac{A_{i}(\bs \lambda)}{z-\lambda_i} \, Y(z, \bs \lambda), 
\end{align*}
taking the poles coalescing to $\lambda = \lambda_1$ as 
$\lambda_j = \lambda + \epsilon \, (j-1) \, s$ for each $j \in \{1,2,\ldots,k+1\}$, 
where $s \in \bC$ with $\Re(s) \geq 0$, 
then again a direct computation shows that the confluence limit $\epsilon \to 0$ gives rise to a non-Fuchsian equation of the form
(see also~\cite[Section~4.2.1]{GMR:Isomonodromic_deformations_confluence_reduction_and_quantisation})
\begin{align*}
\frac{\partial}{\partial z} Y(z,\lambda; s) 
= \bigg( \frac{A_{1,0}(\lambda;s)}{z-\lambda} \; + \; \cdots \; + \; \frac{ \, A_{1,k}(\lambda;s)}{(z-\lambda)^{k+1}} \bigg)
\, Y(z,\lambda; s) ,
\end{align*}
where
\begin{align*}
A_{1,0}(\lambda;s) \coloneqq \; & \lim_{\epsilon \to 0} \sum_{i=1}^{k+1} A_{i}(\bs \lambda) , \\
A_{1,k}(\lambda;s) \coloneqq \; & \lim_{\epsilon \to 0} \epsilon^k \,s^{k}\, \sum_{l=2}^{k+1} (l-1)^k A_{l}(\bs \lambda) , \qquad \textnormal{and where} \quad \lambda_j = \lambda + \epsilon \, (j-1) \, s \textnormal{ for all } j .
\end{align*}
In particular, if the matrix $A_{1,k}(\lambda;s)$ is diagonalizable as in~\eqref{def:diagLij_gen}, 
and its eigenvalues equal $\pm s^{k}$, we see that the highest order Birkhoff invariant in this system is given by $s^k$.

\begin{remark}
As in \Cref{rem:legitimate}, it is possible to choose the matrices in the Fuchsian equation to obtain a desired non-Fuchsian equation: the problem reduces to linear equations
\begin{align*}
\left(
\begin{matrix}
1 & 1 & 1 & \cdots & 1 \\
0 & \varepsilon & 2 \varepsilon & \cdots & k \varepsilon \\
0 & \varepsilon^2 & (2 \varepsilon)^2 & \cdots & (k \varepsilon)^2 \\
\vdots & \vdots & \vdots & \cdots & \vdots \\
0 & \varepsilon^k & (2 \varepsilon)^k & \cdots & (k \varepsilon)^k 
\end{matrix}
\right)
\left(
\begin{matrix}
A_1 \\
A_2 \\
A_3 \\
\vdots \\
A_{k+1}
\end{matrix}
\right)
= 
\left(
\begin{matrix}
A_{1,0} \\
A_{1,1} \\
A_{1,2} \\
\vdots \\
A_{1,k}
\end{matrix}
\right) , \qquad \varepsilon > 0 ,
\end{align*}
which is non-degenerate, 
as the determinant of the matrix can be evaluated using the cofactor expansion with respect to its first column and the Vandermonde determinant, 
\begin{align*}
\left| \;
\begin{matrix}
1 & 1 & 1 & \cdots & 1 \\
0 & \varepsilon & 2 \varepsilon & \cdots & k \varepsilon \\
0 & \varepsilon^2 & (2 \varepsilon)^2 & \cdots & (k \varepsilon)^2 \\
\vdots & \cdots & \cdots & \cdots & \vdots \\
0 & \varepsilon^k & (2 \varepsilon)^k & \cdots & (k \varepsilon)^k 
\end{matrix}
\; \right|
= \Big( \prod_{i=1}^{k} i \varepsilon \Big) 
\Big( \prod_{1 \leq i < j \leq k-1} (j \varepsilon - i \varepsilon) \Big)
\neq 0 .
\end{align*}
\end{remark}

\Cref{thm:Sto-Bir-Gen} now follows using the structure of the simultaneous confluence limit of multiple simple poles discussed above,
which shows that the Birkhoff invariant $s_{i,k}$ associated to the ODE~\eqref{eq:non-Fuchsian} corresponding to the pole $\lambda_i$ of order $k+1$ 
can be obtained from the above confluence procedure with parameter $s = s_{i}$, so that the Birkhoff invariant $s_{i,k} = s_{i}^k$ is the $k$:th power of that parameter, to which \Cref{thm:evol_isomonodromic_times} applies.

\begin{proof}[Proof of \Cref{thm:Sto-Bir-Gen}]
Let $\Lambda_t^i \coloneqq g_t(\lambda_i)$ be the time-evolution of the puncture $\lambda_i$,
and let $\Lambda_t^{i,j} \coloneqq g_t(\lambda_{i,j})$ be the time-evolutions of $\lambda_{i,j} \coloneqq \lambda_i + \epsilon \, (j-1) \, s_i$ 
for each $j \in \{2,3,\ldots,k+1\}$.
In the confluence limit, $S_t^{i} = \smash{\underset{\epsilon \to 0}{\lim}} \, \frac{1}{\epsilon} \big( \Lambda_t^{i,2} - \Lambda_t^i \big)
= g_t'(\lambda_i) s_i$ evolves as in \Cref{thm:evol_isomonodromic_times}, and
\begin{align*}
\partial_t S_t^{i,k} 
= \partial_t  (S_t^{i})^k 
= k (S_t^{i})^{k-1} \big( \partial_t S_t^{i}  \big)
= -\frac{2k (S_t^{i})^k }{(\Lambda_t^i-\Driv_t)^2}
= -\frac{2k S_t^{i,k} }{(\Lambda_t^i-\Driv_t)^2} .
\end{align*}
This shows~\eqref{SLEeq-irrgen} and yields \Cref{thm:Sto-Bir-Gen}.
\end{proof}

%% file: tex-arXiv-v2/3-SLEm2.tex
In this section, we pertain to a construction of local ``irregular martingale observables'' describing the evolution of the non-Fuchsian system under the Loewner flow associated to a chordal $\SLE_\kappa$ curve, whose driving function is $\Driv_t = \sqrt{\kappa} B_t$, a re-scaled Brownian motion. As we will see, this is indeed possible when $\kappa = 4$. 
Again, we will first consider the case of singularities of rank one (i.e., double poles) on the Riemann sphere, and later discuss more general singularities on a high level.
The geometric nature of the martingale is governed by its conformal covariance,
which for double poles involves not only derivatives of the uniformizing map (which arise from the usual conformal covariance of SLE martingale observables, or equivalently, the global conformal Ward identities in CFT),
but also the \emph{pre-Schwarzian} and \emph{Schwarzian} of the uniformizing map. 
We expect the general construction to involve higher Schwarzians, suggesting a richer underlying geometric structure (\Cref{subsec:mgle_higher_poles}).

\subsection{SLEs and stochastic calculus\label{subsec:preli_SLE}}

In order to discuss SLE martingale observables, we first recall the basic definitions concerning
the Schramm-Loewner evolution (SLE) random curves and some tools from stochastic calculus that are needed to derive and understand \Cref{thm:main_mgle} and its variants. 

\subsubsection{Schramm-Loewner evolutions}

Schramm-Loewner evolutions ($\SLE_\kappa$) are a one-parameter family of random conformally invariant curves that arise as scaling limits of interfaces in planar critical models. 
One way to define them is via the Loewner ODE~\eqref{eq:LE}, by taking 
the driving function to be $\Driv_t = \sqrt{\kappa}\, B_t$, with $(B_t)_{t \geq 0}$ a standard Brownian motion (started at $B_0=0$).
For each point $z \in \hat\bC$, its Loewner flow $(g_t(z))_{t \geq 0}$ is defined 
by solving the Loewner equation~\eqref{eq:LE} up to the blow-up time $T_z \coloneqq \sup\{ t \geq 0 \colon \underset{r \in [0,t]}{\inf} \, |g_r(z) - \sqrt{\kappa}\, B_r| > 0 \}$ 
(i.e., when $g_{T_z}(z) = \sqrt{\kappa}\, B_{T_z}$).

The family $\big(K_t \coloneqq \{ z \in \overline{\bH} \colon T_z \leq t \}\big)_{t\geq 0}$ of compact subsets of the upper half-plane $\overline{\bH}$ is called the \emph{chordal Loewner chain} generated by $\Driv$. 
One can show that the maps $g_t$ are, in fact, conformal bijections from $\bH \setminus K_t$ onto $\bH$, and they naturally extend to conformal bijections from the reflected domain $\hat\bC \setminus (K_t \cup K_t^*)$ onto a subset of the sphere $\hat\bC$ (see \Cref{fig:SLE} for an illustration). 
Moreover, they satisfy $g_t(z) = z + \frac{2t}{z} + O(z^{-2})$ as $z \to \infty$.

With $\Driv = \sqrt{\kappa}\, B$, the associated family $(K_t)_{t \geq 0}$ defines the \emph{chordal $\SLE_\kappa$ process} in $\overline{\bH}$ started at $K_0 = \{0\} = \{B_0\}$. 
It is a nontrivial task to show that for each time instant $t > 0$, the chordal $\SLE_\kappa$ process is given by a continuous curve $\gamma$ in the sense that the domain $\bH \setminus K_t$ is the unbounded connected component of $\bH \setminus \gamma[0,t]$~\cite{Rohde-Schramm:Basic_properties_of_SLE}. 
When the variance parameter for the Brownian motion is small enough ($0 \leq \kappa \leq 4$)\footnote{In general, for $\kappa \leq 4$, the curve is almost surely simple; for $4 < \kappa < 8$, it almost surely has self-touchings and touches the boundary $\bR$ of $\bH$; and for $\kappa \geq 8$, it becomes space-filling almost surely.},  
this curve is simple (injective), which is the case of concern in the present article: 
in what follows, we focus exclusively on the critical case $\kappa = 4$. 
See~\cite[Chapter~5]{Kemppainen:SLE_book} and references therein for an extensive introduction to SLE and, in particular, the proofs of its basic properties.

\subsubsection{A glimpse into It\^o calculus}

Throughout this work, we make use of some notions and results from stochastic calculus, which we now briefly recall. 
For a comprehensive introduction to the subject and the definitions and proofs of the statements below, we refer the reader to \cite[Chapter~2]{Durrett:Stochastic_calculus}. 

Let $T$ be a predictable stopping time with respect to a given filtration $\cF_\bullet=(\cF_t)_{t \geq 0}$. 
Throughout, we tacitly restrict to $t \in [0,T)$.
A~\emph{continuous semimartingale} is a stochastic process $(X_t)_{t \geq 0}$ that admits a decomposition
\begin{align*}
X_t = M_t + V_t, \qquad t < T,
\end{align*}
where the process $M=(M_t)_{t \geq 0}$ is a continuous local $\cF_\bullet$-martingale and the process $V=(V_t)_{t \geq 0}$ is a continuous $\cF_\bullet$-adapted process of finite variation. Such a decomposition is unique up to indistinguishability if one assumes $M_0 = 0$ (or, equivalently, $V_0 = 0$). 

For a progressively measurable process $H=(H_t)_{t \geq 0}$ satisfying the usual integrability conditions, the \emph{stochastic integral} of $H$ with respect to $X$ is defined by 
\begin{align*}
\int_0^t H_s \ud X_s \coloneqq \int_0^t H_s \ud M_s + \int_0^t H_s \ud V_s, \qquad t < T,
\end{align*}
where the first term is the It\^o integral with respect to the local martingale $M$, and in particular itself is a local martingale, and the second is the (pathwise) Stieltjes integral with respect to the finite-variation part $V$.

For two continuous semimartingales $X$ and $Y$, their \emph{covariation} (or \emph{quadratic covariation}) process $\langle X, Y\rangle = (\langle X, Y \rangle_t)_{t \geq 0}$ is defined as the limit in probability of the Riemann sums
\begin{align*}
\sum_{k=0}^{n-1} (X_{t_{k+1}} - X_{t_k})(Y_{t_{k+1}} - Y_{t_k}),
\end{align*}
taken along any sequence of deterministic (a priori fixed) partitions ${0=t_0<t_1<\cdots<t_n=t}$ of $[0,t]$ whose mesh tends to zero. The resulting process $\langle X, Y\rangle$ is continuous, adapted, and of finite variation. 
If one of the two processes is of finite variation, $\langle X, Y\rangle$ is necessarily zero. Furthermore, if both $X$ and $Y$ are local martingales, then $XY - \langle X, Y\rangle$ is a local martingale. These properties imply, in particular, that the covariation process is independent of the specific choice of partitions.
As a basic example, $\langle B,B\rangle_t = t$ is the quadratic variation (self-covariation) of the standard Brownian motion $B=(B_t)_{t \geq 0}$ --- a fact that underlies L\'evy's characterization of Brownian motion, frequently used in identifying the $\SLE_\kappa$ driving process $\Driv = \sqrt{\kappa}\, B$.
By bilinearity of the quadratic covariation, we have $\langle \Driv, \Driv\rangle_t = \kappa \, t$.

The fundamental tool needed in the present work is \emph{It\^o's formula}, which is a generalization of the fundamental theorem of calculus for stochastic prosesses, based on a Taylor series expansion including second order corrections for semimartingales of infinite variation. 
We will use it for a vector of semimartingales defined on a common probability space.

\begin{prop}[It\^o's formula]
Let $\bs{X} = (X^1, \ldots, X^d)$ be a $d$-dimensional vector of continuous semimartingales defined on a common probability space, and let $f \in C^{2}(\bR^d)$. Then, its It\^o differential at $\bs{X}$ equals 
\begin{align} \label{eq:Ito_formula_multi}
\ud f(\bs{X}_t) = \sum_{i=1}^d (\partial_{i} f)(\bs{X}_t) \ud X_t^i + \frac{1}{2} \sum_{i,j=1}^d (\partial_{i}\partial_{j} f)(\bs{X}_t) \ud \langle X^i, X^j \rangle_t .
\end{align}
Throughout, such differential relations are understood in the integral sense:
\begin{align*}
f(\bs{X}_t) - f(\bs{X}_0) = \sum_{i=1}^d \int_0^t (\partial_{i} f)(\bs{X}_s) \ud X_s^i  + \frac{1}{2} \sum_{i,j=1}^d \int_0^t (\partial_{i}\partial_{j} f)(\bs{X}_s)\ud \langle X^i, X^j \rangle_s.
\end{align*}
\end{prop}

In what follows, we will often consider vectors consisting of a semimartingale 
\begin{align*}
X^1_t = X^1_0 + c \, \int_0^t \ud B_s + \int_0^t H_s \ud s  ,
\qquad \textnormal{i.e.} \quad
\ud X^1_t = c \ud B_t + H_t \ud t ,
\end{align*}
for some constant $c\in \bR$ and an explicit process $H$ (smooth at small times), 
and processes $(X^i)_{i\geq 2}$ of finite variation for the other components\footnote{For complex-valued processes, we consider their real and imaginary parts as separate components.}. 
Then, Equation~\eqref{eq:Ito_formula_multi} reduces to
\begin{align*}
\ud f(\bs{X}_t) = \; & c \, (\partial_{1} f)(\bs{X}_t) \ud B_t 
+ \frac{c^2}{2} (\partial_{1}^2 f)(\bs{X}_t)\ud t 
+ (\partial_{1} f)(\bs{X}_t) H_t \ud t 
+ \sum_{i=2}^d (\partial_{i} f)(\bs{X}_t)\ud X_t^i .
\end{align*}
This process is a local martingale if and only if the drift term vanishes:
\begin{align*}
0 = \frac{c^2}{2} (\partial_{1}^2 f)(\bs{X}_t)\ud t 
+ (\partial_{1} f)(\bs{X}_t) H_t \ud t 
+ \sum_{i=2}^d (\partial_{i} f)(\bs{X}_t)\ud X_t^i .
\end{align*}

In auxiliary results in \Cref{prop:mgle_beyond_ct,prop:non_BPZ_PDE}, 
which we expect to be of interest to experts for applications to SLE theory, 
we will use a generalization of It\^o's formula for so-called Dirichlet processes~\cite{Fukushima:Dirichlet_forms_and_Markov_processes}, 
which may include parts of infinite variation (but finite $p$-variation for $p\in (1,2)$).
This can be found, e.g., from~\cite[Section~2 and Theorem~3.1]{Bertoin:Sur_une_integrale_pour_processus_variation_bornee}, and uses a generalization of Lebesgue-Stieltjes integrals due to F\"ollmer~\cite{Follmer:Calcul_dIto_sans_probabilites}.

\subsection{$\SLE_4$ martingale observables with double poles\label{subsec:mgle_double_pole}}

Let $\gamma \colon [0,\infty) \to \bH$ be the chordal $\SLE_\kappa$ curve with $\kappa=4$,
whose Loewner driving function is $\Driv = 2B$. 
It is convenient to consider its evolution in the Schottky double $\hat\bC$: the Loewner flow 
$g_t \colon \hat\bC \setminus (\gamma[0,t] \cup \gamma^*[0,t]) \to \hat\bC \setminus [\Xi^-_t,\Xi^+_t]$
is obtained by solving the ODE~\eqref{eq:LE}, and $\Xi^\pm_t=g_t(0\pm) \in \bR$ are the two images of $0 = \gamma(0)$ under the flow (\Cref{fig:SLE}).

Consider now the punctured sphere $\hat\bC_{\bs\lambda} \coloneqq \hat\bC \setminus \{\lambda_1,\ldots,\lambda_n\}$, and
the case where the Poincar\'e rank of each singularity equals $r_i = 1$. 
In~\eqref{eq:non-Fuchsian-DP}, 
let $(s_{i},-s_{i})$ be the eigenvalues of $A_{i,1}(\bs\lambda;\bs s) \in \sl_2(\bC)$ (Birkhoff invariants), and 
$(\alpha_{i},-\alpha_{i})$ the eigenvalues of $A_{i,0}(\bs\lambda;\bs s) \in \sl_2(\bC)$, playing the role of conformal weights in the martingale observable in \Cref{thm:main_mgle}.
Let $\tau(\bs\lambda; \bs s)$ be the tau-function~\eqref{def:tau} and $Y(z, \bs\lambda; \bs s)$ the local solution on $\hat\bC_{\bs\lambda}$ to the IDEs~\eqref{eq:non-Fuchsian-DP}. 

\Mgle*

The proof relies on a computation using It\^o's formula~\eqref{eq:Ito_formula_multi} and, crucially,
the time evolutions of the Birkhoff invariants which we obtained in \Cref{thm:evol_isomonodromic_times}.
Another important fact is that the exponents of formal monodromy, $\alpha_i$, are constants in the isomonodromic flow (\Cref{rem:alpha_const}), 
which implies that they also do not evolve in the Loewner flow.

\begin{proof}
To simplify notation, let us write $M_t$ defined in Equation~\eqref{eq:Mart} in the form
\begin{align}\label{eq:Martingale}
M_t = F(\bs\Lambda_t; \bs S_t) \, \tau(\bs\Lambda_t; \bs S_t) \, Y(\Driv_t, \bs\Lambda_t; \bs S_t)
\eqqcolon F_t \, \tau_t \, Y_t .
\end{align}
Note that both $(F_t)_{t \geq 0}$ and $(\tau_t)_{t \geq 0}$ are, at least up to a stopping time, finite-variation processes,
while $(Y_t)_{t \geq 0}$ is a matrix-valued continuous semimartingale with a nontrivial local martingale part.
Indeed, for $\Driv_t = 2B_t$, the Loewner equation~\eqref{eq:LE} and \Cref{thm:evol_isomonodromic_times} yield the SDEs
\begin{align}\label{eq:Mart_rules1}
\ud \Driv_t = 2\ud B_t , \qquad
\ud \Lambda^i_t = \frac{2 \ud t}{\Lambda^i_t - \Driv_t} , \qquad
\ud S_t^{i} = - \frac{2 S_t^{i}\, \ud t}{(\Lambda_t^i - \Driv_t)^2} ,
\end{align}
and $\ud \langle \Driv,\Driv \rangle_t = 4 \ud t$. 
Using It\^o's formula~\eqref{eq:Ito_formula_multi} and the linear equations~(\ref{eq:non-Fuchsian-DP},~\ref{eq:LaxUi},~\ref{eq:LaxVi}), 
because only the quadratic variation of $\Driv$ contributes to the second order terms,
we have
\begin{align*}
\ud Y_t = \; & A(\Driv_t, \bs\Lambda_t; \bs S_t) \, Y_t \ud \Driv_t 
\; + \;  \sum_{i=1}^n U_i(\Driv_t, \bs\Lambda_t; \bs S_t) \, Y_t \ud \Lambda^i_t 
+ \sum_{i=1}^n V_i(\Driv_t, \bs\Lambda_t; \bs S_t) \, Y_t \ud S^i_t \\
\; & \qquad\qquad\qquad\qquad\;\,
\; + \; \frac{1}{2} \underbrace{\Big(\big( \tfrac{\partial}{\partial z} A\big)(z, \bs\Lambda_t; \bs S_t)\Big|_{z=\Driv_t} 
\; + \; \big(A(\Driv_t, \bs\Lambda_t; \bs S_t) \big)^2 \Big) \, Y_t}_{= \; \big( \tfrac{\partial^2}{\partial z^2} Y\big)(z, \bs\Lambda_t; \bs S_t)|_{z=\Driv_t} }  \underbrace{\ud \langle \Driv,\Driv \rangle_t}_{= \; 4 \ud t} 
\\
= \; & 2 A_t \, Y_t \ud B_t 
\; + \; 2 \Bigg( \sum_{i=1}^n \frac{ U_t^i}{\Lambda^i_t - \Driv_t}
\; - \; \sum_{i=1}^n \frac{S_t^{i} \, V_t^i}{(\Lambda_t^i - \Driv_t)^2} 
\; + \; \Big( \tfrac{\partial}{\partial z} A_t + \tfrac{1}{2} \Tr(A_t^2) \, \mathbf{1} \Big) \Bigg) \, Y_t \ud t ,
\end{align*}
writing, for ease, $A_t = A(\Driv_t, \bs\Lambda_t; \bs S_t)$ and $U_t^i = U_i(\Driv_t, \bs\Lambda_t; \bs S_t)$ and $V_t^i = V_i(\Driv_t, \bs\Lambda_t; \bs S_t)$ 
and $\tfrac{\partial}{\partial z} A_t = \big(\tfrac{\partial}{\partial z} A\big)(z, \bs\Lambda_t; \bs S_t)\big|_{z=\Driv_t}$. 
From~\eqref{eq:non-Fuchsian-DP}, we see that the latter equals
\begin{align*}
\frac{\partial}{\partial z} A_t = 
- \sum_{i=1}^{n} \bigg( \frac{A_{i,0}(\bs\Lambda_t; \bs S_t)}{(\Lambda^i_t - \Driv_t)^2} - 2 \, \frac{A_{i,1}(\bs\Lambda_t; \bs S_t)}{(\Lambda_t^i - \Driv_t)^{3}} \bigg) ,
\end{align*}
and from~(\ref{eq:LaxUi},~\ref{eq:LaxVi}), we have
\begin{align*}
U_t^{i} = - \frac{A_{i,1}(\bs\Lambda_t; \bs S_t)}{(\Lambda^i_t - \Driv_t)^2} + \frac{A_{i,0}(\bs\Lambda_t; \bs S_t)}{\Lambda^i_t - \Driv_t} 
\qquad \textnormal{and} \qquad
S_t^{i}V_t^{i} = \frac{A_{i,1}(\bs\Lambda_t; \bs S_t)}{(\Lambda^i_t - \Driv_t)} .
\end{align*}
Hence, the drift term of $Y_t$ simply equals $\Tr(A_t^2) \, Y_t \ud t$:
\begin{align}\label{eq:dY_final_expression}
\ud Y_t = \; & \big( 2 A_t \ud B_t + \Tr(A_t^2) \ud t \big) \, Y_t .
\end{align}
Overall, since the covariation process $\langle F \, \tau,Y \rangle$ is zero and the processes $\tau$ and $F$ are scalar-valued, we find that
\begin{align} \label{eq:dM}
\begin{split}
\ud M_t = \; & F_t \, \tau_t \ud Y_t + (\tau_t \ud F_t + F_t \, \tau_t \ud \log \tau_t) \, Y_t 
\\
= \; & \Big( 2 A_t \ud B_t + \Tr(A_t^2) \ud t  + \ud \log F_t + \ud \log \tau_t \Big) M_t .
\end{split}
\end{align}
Recall now from \Cref{lem:Hamiltonians} that the quantity $\Tr(A_t^2)$ has an explicit expression~\eqref{eq:TrL2}:
\begin{align}\label{eq:TrL2_again}
\Tr \big(A_t^2\big) 
= \sum_{i=1}^{n} \sum_{k=1}^{4} (-1)^{k} \frac{\ell_{i,k}(\bs \Lambda_t; \bs S_t)}{(\Lambda_t^i - \Driv_t)^{k}} ,
\end{align}
and $\ud \log \tau_t$ is obtained from its definition~\eqref{def:tau}
and the identities~(\ref{eqell:Hlamb},~\ref{eqell:Hsgen}) in \Cref{lem:Hamiltonians}:
\begin{align}
\nonumber \ud \log \tau_t 
\overset{\hphantom{\textnormal{\eqref{eqell:Hlamb}}}}{\underset{\hphantom{\textnormal{\eqref{eqell:Hsgen}}}}{=}} 
\; & \sum_{i=1}^n H_{\Lambda^i_t}(\bs \Lambda_t; \bs S_t) \ud \Lambda^i_t 
\; + \; \sum_{i=1}^n H_{S^i_t}(\bs \Lambda_t; \bs S_t) \ud S^i_t 
\\
\overset{\textnormal{\eqref{eqell:Hlamb}}}{\underset{\textnormal{\eqref{eqell:Hsgen}}}{=}}
\; & \frac{1}{2} \, \sum_{i=1}^n \ell_{i,1}(\bs \Lambda_t; \bs S_t) \ud \Lambda^i_t 
\; + \; \sum_{i=1}^n \bigg( \frac{\ell_{i,2}(\bs \Lambda_t; \bs S_t)}{2 S_t^i} \; - \; \frac{(\ell_{i,3}(\bs \Lambda_t; \bs S_t))^2}{16 (S_t^i)^3} \bigg) \ud S^i_t \nonumber
\\
\overset{\textnormal{\eqref{eq:Mart_rules1}}}{\underset{\hphantom{\textnormal{\eqref{eqell:Hsgen}}}}{=}} 
 \; & \sum_{i=1}^n \Bigg( 
\frac{\ell_{i,1}(\bs \Lambda_t; \bs S_t)}{\Lambda^i_t - \Driv_t}
\; - \; \frac{\ell_{i,2}(\bs \Lambda_t; \bs S_t)}{(\Lambda_t^i - \Driv_t)^2}
\; + \; \frac{(\ell_{i,3}(\bs \Lambda_t; \bs S_t))^2}{8 (S_t^i)^2 \, (\Lambda_t^i - \Driv_t)^2} \Bigg) \ud t . \label{eq:Ito_tau}
\end{align}
Here, the terms involving $\ell_{i,1}$ and $\ell_{i,2}$ cancel with those in $\Tr \big(A_t^2\big)$. 
For $\ud \log F_t$, using the Loewner equation~\eqref{eq:LE}, we obtain the variations\footnote{Here, $\PreSchwarzian (g) \coloneqq g''/g'$ is the pre-Schwarzian and $\Schwarzian(g) \coloneqq (g''/g')' - \frac12 (g''/g')^2$ the Schwarzian.}
\begin{align}
\begin{split}
\ud g_t'(\lambda_i) 
= \; & - \frac{2 \, g_t'(\lambda_i)}{(\Lambda_t^i - \Driv_t)^2} \ud t  ,
\\
\ud \PreSchwarzian(g_t)(\lambda_i) 
= \; & \phantom{-} \frac{4 \, g_t'(\lambda_i)}{(\Lambda_t^i - \Driv_t)^3} \ud t  ,
\\
\ud \Schwarzian(g_t)(\lambda_i)
= \; & -\frac{12 \, (g_t'(\lambda_i))^2}{(\Lambda_t^i - \Driv_t)^4} \ud t  .
\end{split} \label{eq:time_space_deriv_SLE}
\end{align}
Combining this with the time-evolutions $S_t^{i} = g_t'(\lambda_i) s_i$ of the Birkhoff invariants from Equation~\eqref{eqn:evol_isomonodromic_times} in \Cref{thm:evol_isomonodromic_times}, 
the crucial fact from \Cref{rem:alpha_const} that $\alpha_i$ are constants in the isomonodromic (and thus Loewner) flow, 
and identities from \Cref{lem:Hamiltonians}, we obtain
\begin{align} 
\nonumber
\ud \log F_t
\overset{\hphantom{\textnormal{\eqref{eq:ellval}}}}{=} \; & \sum_{i=1}^{n} \bigg( 
\alpha_i^2 \ud \log (g_t'(\lambda_i)) 
\; + \; \frac{s_i^2}{6} \ud \Schwarzian(g_t)(\lambda_i)
\; + \; s_i \alpha_i \ud \PreSchwarzian(g_t)(\lambda_i) 
\bigg) \\
\label{eq:Ito_F}
\overset{\hphantom{\textnormal{\eqref{eq:ellval}}}}{=} \; & - 2 \sum_{i=1}^{n} \bigg( 
\frac{\alpha_i^2}{(\Lambda_t^i - \Driv_t)^2}
\; + \; \frac{ (S_t^i)^2 }{(\Lambda_t^i - \Driv_t)^4}
\; - \; \frac{2 \alpha_i S_t^i }{(\Lambda_t^i - \Driv_t)^3} 
\bigg) \ud t \\
\nonumber
\overset{\textnormal{\eqref{eq:ellval}}}{=} \; & - 2 \sum_{i=1}^{n} \Bigg( 
\frac{\big(\ell_{i,3}(\bs \Lambda_t; \bs S_t) / 4 S_t^i \big)^2}{(\Lambda_t^i - \Driv_t)^2}
\; + \; \frac{ \tfrac{1}{2}  \, \ell_{i,4}(\bs \Lambda_t; \bs S_t) }{(\Lambda_t^i - \Driv_t)^4}
\; - \; \frac{ \tfrac{1}{2}  \, \ell_{i,3}(\bs \Lambda_t; \bs S_t)}{(\Lambda_t^i - \Driv_t)^3} 
\Bigg) \ud t .
\end{align}
Observe that the linear terms involving $\ell_{i,4}$ and $\ell_{i,3}$ cancel with those in $\Tr \big(A_t^2\big)$ in~\eqref{eq:TrL2_again},
and the quadratic term involving $\ell_{i,3}^2$ cancels with that in $\ud \log \tau_t$ in~\eqref{eq:Ito_tau}.
In summary, the drift term in $\ud M_t$ in~\eqref{eq:dM} equals zero, so $M$ is indeed a local martingale. 
\end{proof}

One can define natural extensions of $\SLE_\kappa$ by introducing logarithmic drift terms for the driving function. 
These variants describe interfaces subject to boundary condition changes, and encode various restriction and branching properties of SLE.
Such a variant of interest to us is the family $\SLE_\kappa^{\langle\mu\rangle}(\rho)$ (with $\mu = 0$ if $\rho \neq -2$)\footnote{Throughout, we drop the superscript $\mu = 0$ if it equals zero or if it is irrelevant.} with parameters $\rho, \mu \in \bR$. See~\cite{Schramm-Wilson:SLE_coordinate_changes,Miller-Sheffield:Imaginary_geometry1} for details on the case $\mu=0$, and~\cite{Sheffield:Exploration_trees_and_CLEs, Werner-Wu:On_conformally_invariant_CLE_explorations, Lehmkuehler:Trunks_of_CLE4_explorations} 
for the general construction and properties of $\SLE_\kappa^{\langle\mu\rangle}(\rho)$ (see also~\cite[Section~3]{MSW:CLE_percolations} for an overview).

We focus on the case $\kappa = 4$, to which our martingale result applies. 
First, let $\rho \in \bR$ and $\mu = 0$. 
The $\SLE_4(\rho)$ process in $(\bH;0,\infty)$ with force point $\spec \in \bR \setminus \{0\}$ is defined as 
the curve $\gamma$ such that the conformal bijections $g_t \colon \bH \setminus \gamma[0,t] \to \bH$ solve the Loewner equation~\eqref{eq:LE} with driving function $\Driv = (\Driv_t)_{t \geq 0}$ and force point $\Spec = (\Spec_t)_{t \geq 0}$ that solve the SDEs
\begin{align}\label{eq:SLE_kappa(rho)_driving.fct}
\begin{cases}
\displaystyle
\ud \Driv_t = 2 \ud B_t + \ud \Spec_t , \\
\displaystyle
\ud \Spec_t = \frac{- \rho \ud t}{\Spec_t - \Driv_t} , \\
(\Driv_0, \Spec_0) = (0, \spec) ,
\end{cases}
\end{align}
up to the blow-up time (continuation threshold)
\begin{align} \label{eq:continuation_threshold}
T_\spec \coloneqq \sup \Big\{ t \geq 0 \colon \underset{s\in [0,t]}{\inf} \, |\Spec_s - \Driv_s| > 0 \Big\} .
\end{align} 
Observe that when $\rho=-2$, the process $\Spec_t = g_t(\spec)$ evolves exactly according to the Loewner differential equation~\eqref{eq:LE}, and by sending $\spec$ to infinity we recover the chordal $\SLE_4$.
After time $T_\spec$, if $\rho \leq -2$, the process is ill-defined. Suitable generalizations were proposed by Sheffield in~\cite{Sheffield:Exploration_trees_and_CLEs}.
We will recall shortly the special case $\SLE_4^{\langle\mu\rangle}(-2)$, which is strongly motivated by the fact that it traces the loops of an $\CLE_4$ ensemble~\cite{Werner-Wu:On_conformally_invariant_CLE_explorations}.

As before, let $\bs\Lambda_t = (\Lambda_t^1, \ldots,\Lambda_t^{n})$ be the time-evolutions $\Lambda_t^i = g_t(\lambda_i)$ of the punctures 
and $\bs S_t = (S_t^1, \ldots,S_t^{n})$ the time-evolutions $S_t^i = \hat{g}_t(s_i)$ of the associated Birkhoff invariants,
and denote the covariance factor, with $(\alpha_{i},-\alpha_{i})$ the eigenvalues of $A_{i,0}(\bs\lambda;\bs s) \in \sl_2(\bC)$, by
\begin{align} \label{eq:covariance_factor_F}
F(\bs\Lambda_t; \bs S_t) 
\coloneqq \; & \prod_{i=1}^{n} (g_t'(\lambda_i))^{\alpha_i^2} \exp \bigg(\frac{s_i^2}{6} \, \Schwarzian(g_t)(\lambda_i) + s_i \alpha_i \, \PreSchwarzian(g_t)(\lambda_i) \bigg) .
\end{align} 

\begin{thm}
\label{thm:SLE_4(-2)_mart}
Let $(g_t)_{t \geq 0}$ be the Loewner flow corresponding to a chordal $\SLE_4(-2)$ in $(\bH, 0, \infty)$ with force point $\spec \in \bR \setminus \{0\}$. 
Then, the following processes are local martingales. 
\begin{enumerate}
\item \label{item:SLE_4(-2)_mart_mtrx}
The \textnormal{(}matrix-valued\textnormal{)} process
\begin{align*}
F(\bs\Lambda_t; \bs S_t) \, \tau(\bs\Lambda_t; \bs S_t) \, Y^{-1}(\Spec_t, \bs\Lambda_t; \bs S_t) \,  \, Y(\Driv_t, \bs\Lambda_t; \bs S_t) ,
\end{align*}
where $\Driv_t = 2 B_t + \Spec_t$ and $\Spec_t = g_t(\spec)$ are the processes in~\eqref{eq:SLE_kappa(rho)_driving.fct} with $\rho=-2$.

\item \label{item:SLE_4(-2)_mart_scal}
The \textnormal{(}scalar-valued\textnormal{)} process
\begin{align*}
F(\bs\Lambda_t; \bs S_t) \, \tau(\bs\Lambda_t; \bs S_t) \, \Tr \Big( Y^{-1}(\Spec_t, \bs\Lambda_t; \bs S_t) \, Y(\Driv_t, \bs\Lambda_t; \bs S_t) \Big) .
\end{align*}
\end{enumerate}
\end{thm}

\begin{proof}
Consider first (similarly as before) $M_t \coloneqq F_t \, \tau_t \, \widetilde Y_t^{-1} \, Y_t$, writing $\widetilde Y_t \coloneqq Y(\Spec_t, \bs \Lambda_t; \bs S_t)$.
Note that $\widetilde Y_t$ has vanishing quadratic variation: indeed, by It\^o's formula~\eqref{eq:Ito_formula_multi},
\begin{align*}
\ud \widetilde Y_t^{-1} 
= -\widetilde Y_t^{-1} \bigg( A(\Spec_t, \bs \Lambda_t; \bs S_t) \ud \Spec_t 
\; + \; \sum_{i=1}^n U_i(\Spec_t, \bs \Lambda_t; \bs S_t) \ud \Lambda_t^i 
\; + \; \sum_{i=1}^n V_i(\Spec_t, \bs \Lambda_t; \bs S_t) \ud S^i_t \bigg).
\end{align*}
Similarly as in the derivation of~\eqref{eq:dY_final_expression}, adjusted accordingly to~\eqref{eq:SLE_kappa(rho)_driving.fct}, we obtain
\begin{align*}
\ud Y_t = \Big( 2 A_t \ud B_t - \frac{2 A_t}{\Driv_t - \Spec_t} \ud t + \Tr (A_t^2) \ud t \Big) Y_t .
\end{align*}
Combined with~(\ref{eq:LaxUi},~\ref{eq:LaxVi},~\ref{eq:Mart_rules1},~\ref{eq:SLE_kappa(rho)_driving.fct}), these formulas imply that 
\begin{align*}
\ud \big(\widetilde Y_t^{-1} Y_t\big) 
= \; & 2\widetilde Y_t^{-1} A_t Y_t \ud B_t \; + \; \Tr (A_t^2) \widetilde Y_t^{-1} Y_t \ud t \\
\, & - 2\widetilde Y_t^{-1} \sum_{i=1}^n \sum_{k=0}^1 \Bigg( 
\frac{A_{i,k}(\bs\Lambda_t; \bs S_t)}{(\Spec_t - \Lambda^i_t)^{k+1} (\Spec_t - \Driv_t)} 
\; + \; \frac{A_{i,k}(\bs\Lambda_t; \bs S_t)}{(\Spec_t - \Lambda^i_t)^{k+1} (\Driv_t - \Lambda^i_t )} \\
\; & \phantom{-2\widetilde Y_t^{-1} \sum_{i=1}^n \sum_{k=1}^2 \Bigg(} 
\; + \; \frac{A_{i,1}(\bs\Lambda_t; \bs S_t) \, \one\{k = 1\}}{(\Spec_t - \Lambda^i_t) (\Driv_t - \Lambda^i_t)^2} 
\; + \; \frac{A_{i,k}(\bs\Lambda_t; \bs S_t)}{(\Driv_t - \Lambda^i_t)^{k+1} (\Driv_t - \Spec_t)}\Bigg) Y_t \ud t\\
= \; & 2\widetilde Y_t^{-1} A_t Y_t \ud B_t - (\ud \log F_t + \ud \log \tau_t) \, \widetilde Y_t^{-1} Y_t \ud t .
\end{align*}
Hence, we find that $M$ has zero drift, so it is a local martingale: 
\begin{align*}
\ud M_t 
= F_t \tau_t \Big( \ud (\widetilde Y_t^{-1}Y_t) + (\ud \log F_t + \ud \log \tau_t) \, \widetilde Y_t^{-1}Y_t \ud t \Big) 
= 2 F_t \tau_t \widetilde Y_t^{-1} A_t Y_t \ud B_t .
\end{align*}
This proves Item~\ref{item:SLE_4(-2)_mart_mtrx}. 
For Item~\ref{item:SLE_4(-2)_mart_scal}, we can just take the trace of both sides, which yields
\begin{align*}
\ud \Tr (M_t) = 2 F_t \tau_t \Tr \big(\widetilde Y_t^{-1} A_t Y_t \big) \ud B_t 
= 2 \, \Tr \big( A_t M_t \big) \ud B_t .
\end{align*}
This concludes the proof. 
\end{proof}

The processes in \Cref{thm:SLE_4(-2)_mart} are well-defined up to the stopping time $T_\spec \wedge T_{\bs\Lambda}$,
that is, the minimum of the continuation threshold $T_\spec$ defined in~\eqref{eq:continuation_threshold} and the first hitting time $T_{\bs\Lambda}$ of $\Driv$ or $\Xi$ to any of the marked points $\Lambda_s^i$, $i\in \{1, \ldots, n\}$ (see~\eqref{eq:swallow_times}). 
After the continuation threshold $T_\spec$, the $\SLE_4(-2)$ process can be continued as an $\SLE_4^{\langle\mu\rangle}(-2)$ process,
where $\mu \in \bR$ is a parameter. 
Namely, one can define the process $\Xi^{\langle \mu\rangle}$ as
\begin{align} \label{eq:I^0_def}
\Xi^{\langle \mu\rangle}_t \coloneqq \; & \mu \, \ell^0_t - \textnormal{p.v.} \int_0^t \frac{\ud s}{B_s} \eqqcolon \mu \, \ell^0_t - I^0_t , \qquad \mu \in \bR ,
\end{align}
where $\ell^0$ is the local time of the standard Brownian motion $B$ at $0$ and ``$\textnormal{p.v.} \int$'' stands for the principal value integral\footnote{For the precise definition of the principal value $I^0$, see~\cite[Propositions~3.2~\&~3.4]{Lehmkuehler:Trunks_of_CLE4_explorations} (our $I^0$ here corresponds to $I(1,0)$ in the reference) or~\cite[Proposition~3.8]{Sheffield:Exploration_trees_and_CLEs}.};
almost surely, $\ud I^0_t - \frac{\ud t}{B_t}$ is zero on any time-interval where $B$ is non-zero. 
One can prove that the $\SLE_4^{\langle\mu\rangle}(-2)$ process defined using the Loewner equation~\eqref{eq:LE} with 
driving function $\Driv^{\langle \mu\rangle} = 2 B + \Xi^{\langle \mu\rangle}$ is almost surely generated by a continuous curve
(see~\cite[Proposition~5.3]{MSW:CLE_percolations}). 
Note that the value of $\mu$ only affects the evolution of $(\Driv^{\langle \mu\rangle}_t,\Xi^{\langle \mu\rangle}_t)$ at times $t$ when $\Driv^{\langle\mu\rangle}_t = \Xi^{\langle\mu\rangle}_t$, which form the so-called \emph{trunk} of the exploration process.
By~\cite[Proposition~4]{Werner-Wu:On_conformally_invariant_CLE_explorations}, the $\CLE_4$ loops traced by the exploration process are independent of the choice of $\mu$, and by~\cite{MSW:CLE_percolations,Lehmkuehler:Trunks_of_CLE4_explorations}, 
the trunk has the law of an $\SLE_4(\rho,-2-\rho)$ process with $\rho = \frac{2}{\pi}\, \mathrm{arccot}( -\mu/\pi )$ and two force points.

\begin{prop} \label{prop:mgle_beyond_ct}
Let $(g_t)_{t \geq 0}$ be the Loewner flow corresponding to a chordal $\SLE_4^{\langle\mu\rangle}(-2)$ in $(\bH, 0, \infty)$ with force point $\Xi^{\langle \mu\rangle}_0 = 0$. 
The \textnormal{(}scalar-valued\textnormal{)} process in Item~\ref{item:SLE_4(-2)_mart_scal} of \Cref{thm:SLE_4(-2)_mart},
\begin{align*}
N_t \coloneqq 
F(\bs\Lambda_t; \bs S_t) \, \tau(\bs\Lambda_t; \bs S_t) \, \Tr \Big( Y^{-1}(\Xi^{\langle\mu\rangle}_t, \bs\Lambda_t; \bs S_t) \, Y(\Driv^{\langle\mu\rangle}_t, \bs\Lambda_t; \bs S_t) \Big)  , 
\end{align*}
is well-defined up to the stopping time 
\begin{align}\label{eq:swallow_times}
T_{\bs\Lambda} \coloneqq \; & \sup \Big\{t\geq 0 \colon \inf_{s\in [0,t]} |\Lambda_s^i - \Driv^{\langle\mu\rangle}_s| > 0 \textnormal{ and } \min_{i\in \{1, \ldots, n\}}\inf_{s\in [0,t]} |\Lambda_s^i - \Xi^{\langle\mu\rangle}_s| > 0 \Big\} .
\end{align}
\end{prop} 

\begin{proof}
Write $N_t = F_t \, \tau_t \, \Tr \big( \widetilde Y_t^{-1} \, Y_t \big)$,
where $Y_t \coloneqq Y(\Driv^{\langle\mu\rangle}_t, \bs\Lambda_t; \bs S_t)$ 
and $\widetilde Y_t \coloneqq Y(O^{\langle\mu\rangle}_t, \bs \Lambda_t; \bs S_t)$, 
and write $A_t \coloneqq A(\Driv^{\langle\mu\rangle}_t, \bs \Lambda_t; \bs S_t)$ and $\widetilde A_t \coloneqq A(\Xi^{\langle\mu\rangle}_t, \bs \Lambda_t; \bs S_t)$. 
By the IDEs~\eqref{eq:non-Fuchsian-DP}, we have
\begin{align*}
\tfrac{\partial}{\partial z} Y_t 
\coloneqq \; & \tfrac{\partial}{\partial z} Y(z, \bs\Lambda_t; \bs S_t) \big|_{z=\Driv^{\langle \mu\rangle}_t}
\;\;\;\;\, = \; \phantom{-} A(\Driv^{\langle \mu\rangle}_t, \bs\Lambda_t; \bs S_t) \, Y(\Driv^{\langle \mu\rangle}_t, \bs\Lambda_t; \bs S_t) 
&& \eqqcolon A_t Y_t , 
\\
\tfrac{\partial}{\partial \spec} \widetilde Y_t^{-1}
\coloneqq \; &
\tfrac{\partial}{\partial \spec} Y^{-1}(\xi, \bs\Lambda_t; \bs S_t) \big|_{\xi=\Xi^{\langle \mu\rangle}_t}
\; = \; - Y^{-1}(\Xi^{\langle \mu\rangle}_t, \bs\Lambda_t; \bs S_t) \, A(\Xi^{\langle \mu\rangle}_t, \bs\Lambda_t; \bs S_t)
&& \eqqcolon - \widetilde Y_t^{-1} \widetilde A_t .
\end{align*}
Since the processes $\Driv^{\langle \mu\rangle}$ and $\Xi^{\langle \mu\rangle}$ are not continuous semimartingales, 
instead of~\eqref{eq:Ito_formula_multi} 
we use a version of It\^o's formula for Dirichlet processes (with zero-energy summand of finite $p$-variation for $p\in (1,2)$) from~\cite[Section~2 and Theorem~3.1]{Bertoin:Sur_une_integrale_pour_processus_variation_bornee}.
Like in the derivation of~\eqref{eq:dY_final_expression}, now adjusted accordingly to 
$\Driv^{\langle \mu\rangle} = 2 B + \Xi^{\langle \mu\rangle}$ and~\eqref{eq:I^0_def} for $\Xi^{\langle \mu\rangle}$, we obtain 
\begin{align*}
\ud Y_t = \;\; & \Big( 2 A_t \ud B_t + \Tr (A_t^2) \ud t \Big) Y_t \; + \; (\tfrac{\partial}{\partial z} Y_t) \ud \Xi^{\langle\mu\rangle}_t ,
\\
\ud \widetilde Y_t^{-1} 
= \;\; &  (\tfrac{\partial}{\partial \spec} \widetilde Y_t^{-1}) \ud \Xi^{\langle\mu\rangle}_t 
\; - \; \widetilde Y_t^{-1} \bigg( \sum_{i=1}^n U_i(\Xi^{\langle\mu\rangle}_t, \bs \Lambda_t; \bs S_t) \ud \Lambda_t^i 
\; + \; \sum_{i=1}^n V_i(\Xi^{\langle\mu\rangle}_t, \bs \Lambda_t; \bs S_t) \ud S^i_t \bigg).
\end{align*}
Hence, using~(\ref{eq:LaxUi},~\ref{eq:LaxVi},~\ref{eq:Mart_rules1}) similarly as before, we find 
\begin{align*}
\ud \big(\widetilde Y_t^{-1} Y_t\big) 
= \; & 2\widetilde Y_t^{-1} A_t Y_t \ud B_t
\; + \; \Big( (\tfrac{\partial}{\partial z} + \tfrac{\partial}{\partial \spec})(\widetilde Y_t^{-1} Y_t) \Big) \ud \Xi^{\langle\mu\rangle}_t 
\; + \; \Tr (A_t^2) \widetilde Y_t^{-1} Y_t \ud t \\
\, & - \widetilde Y_t^{-1} \sum_{i=1}^n \sum_{k=0}^1 \Bigg( 
\frac{2 A_{i,k}(\bs\Lambda_t; \bs S_t)}{(\Xi^{\langle \mu\rangle}_t - \Lambda^i_t)^{k+1} (\Driv^{\langle \mu\rangle}_t - \Lambda^i_t )}
\; + \; \frac{2 A_{i,1}(\bs\Lambda_t; \bs S_t) \, \one\{k = 1\}}{(\Xi^{\langle \mu\rangle}_t - \Lambda^i_t) (\Driv^{\langle \mu\rangle}_t - \Lambda^i_t)^2}
\Bigg) Y_t \ud t ,
\end{align*}
where $\big(\tfrac{\partial}{\partial z} + \tfrac{\partial}{\partial \spec}\big)(\widetilde Y_t^{-1} Y_t) = \widetilde Y_t^{-1} (\tfrac{\partial}{\partial z} Y_t) + (\tfrac{\partial}{\partial \spec} \widetilde Y_t^{-1}) Y_t = \widetilde Y_t^{-1} A_t Y_t - \widetilde Y_t^{-1} \widetilde A_t Y_t$.
Note that at any time $t$ such that $\Driv^{\langle \mu\rangle}_t = \Xi^{\langle \mu\rangle}_t$, after taking the trace, 
the contribution to $\ud \Xi^{\langle\mu\rangle}_t$ vanishes:
\begin{align*}
\tfrac{\partial}{\partial \spec} \Tr \big( \widetilde Y_t^{-1} \, Y_t \big) 
\overset{\textnormal{\eqref{eq:non-Fuchsian-DP}}}{=} \; & - \Tr \big( \widetilde Y_t^{-1} \widetilde A_t \, Y_t \big) 
\\
\overset{\hphantom{\textnormal{\eqref{eq:non-Fuchsian-DP}}}}{=} \; & - \Tr \big( Y_t^{-1} \widetilde A_t \, Y_t \big) 
= \Tr( \widetilde A_t ) 
&& \textnormal{[since $\Driv^{\langle \mu\rangle}_t = \Xi^{\langle \mu\rangle}_t$]}
\\
\overset{\hphantom{\textnormal{\eqref{eq:non-Fuchsian-DP}}}}{=} \; & 0 = \Tr( A_t )
= \Tr \big( \widetilde Y_t^{-1} A_t \, \widetilde Y_t \big) 
&& \textnormal{[since $\widetilde A_t = A_t \in \sl_2(\bC)$]}
\\
\overset{\hphantom{\textnormal{\eqref{eq:non-Fuchsian-DP}}}}{=} \; & \Tr \big( \widetilde Y_t^{-1} A_t \, Y_t \big) 
&& \textnormal{[since $\Driv^{\langle \mu\rangle}_t = \Xi^{\langle \mu\rangle}_t$]}
\\
\overset{\hphantom{\textnormal{\eqref{eq:non-Fuchsian-DP}}}}{=} \; & \tfrac{\partial}{\partial z} \Tr \big( \widetilde Y_t^{-1} \, Y_t \big)  .
&& \textnormal{[by the IDEs~\eqref{eq:non-Fuchsian-DP}]}
\end{align*}
Therefore, recalling also that $2 B = \Driv^{\langle \mu\rangle} - \Xi^{\langle \mu\rangle}$ and using~\eqref{eq:I^0_def}, we find that the term involving the local time $\ell^0$ of $B$ vanishes,
and the principal value integral in~\eqref{eq:I^0_def} thus yields the convergent integral 
\begin{align*} 
\; & \int_0^t \Tr \Big( \big(\tfrac{\partial}{\partial z} + \tfrac{\partial}{\partial \spec}\big)(\widetilde Y_s^{-1} Y_s) \Big) \ud \Xi^{\langle \mu\rangle}_s
\\
= \; & \mu \, \int_0^t  \underbrace{\Tr \Big( \big(\tfrac{\partial}{\partial z} + \tfrac{\partial}{\partial \spec}\big)(\widetilde Y_s^{-1} Y_s) \Big)}_{= \; 0 \textnormal{ on the support of $\ell^0$}} \ud \ell^0_s
\; - \; 2 \, \textnormal{p.v.} \int_0^t \frac{\Tr \Big( \big(\tfrac{\partial}{\partial z} + \tfrac{\partial}{\partial \spec}\big)(\widetilde Y_s^{-1} Y_s) \Big)\ud s}{ \Driv^{\langle \mu\rangle}_s - \Xi^{\langle \mu\rangle}_s} 
\\
= \; &  -2 \int_0^t \frac{\Tr \Big( \big(\tfrac{\partial}{\partial z} + \tfrac{\partial}{\partial \spec}\big)(\widetilde Y_s^{-1} Y_s) \Big)\ud s}{ \Driv^{\langle \mu\rangle}_s - \Xi^{\langle \mu\rangle}_s} 
\; = \;  -2 \int_0^t \frac{\Tr \Big( Y_s \widetilde Y_s^{-1} (A_s - \widetilde A_s) \Big)\ud s}{ \Driv^{\langle \mu\rangle}_s - \Xi^{\langle \mu\rangle}_s}.
\end{align*}
Recalling that $\Tr (A_t^2) \widetilde Y_t^{-1} Y_t \ud t = - (\ud \log F_t + \ud \log \tau_t) \, \widetilde Y_t^{-1} Y_t \ud t$,
with a similar computation as in the proof of \Cref{thm:SLE_4(-2)_mart}, we finally conclude that
\begin{align*}
\ud \Tr \big(\widetilde Y_t^{-1} Y_t\big) 
= \; & 2 \, \Tr \big(\widetilde Y_t^{-1} A_t Y_t\big) \ud B_t - (\ud \log F_t + \ud \log \tau_t) \, \Tr \big(\widetilde Y_t^{-1} Y_t \big) \ud t .
\end{align*}
Hence, we find that $N$ is well-defined and has zero drift up to the stopping time $T_{\bs\Lambda}$.
\end{proof}

\subsection{Confluent BPZ partial differential equations\label{subsec:BPZ}}

Martingale observables encode conserved quantities for the SLE flow, and are closely related to special representations of the Virasoro algebra~\cite{BPZ:Infinite_conformal_symmetry_in_2D_QFT, Bauer-Bernard:SLE_martingales_and_Virasoro_algebra, Bauer-Bernard:Conformal_transformations_and_SLE_partition_function_martingale}. 
In particular, the 2nd order BPZ PDEs satisfied by correlation functions of CFT degenerate fields at level two are equivalent to the vanishing of the It\^o drifts in the associated martingale observables (for more details,  see~\cite{Peltola:Towards_CFT_for_SLEs} and references therein).  
Such an association generalizes to irregular martingales as well. 
Indeed, in this section we show that the irregular $\SLE_{4}$ martingale solves a confluent BPZ equation with central charge $c=1$, and prove \Cref{thm:BPZ_general}.

\bigskip 

Let $\FunSpGen(\cXFP_n^{1})$ denote functions on the configuration space $\cXFP_n^{1} \subset \bR \times \bR \times \bC^{2n} \cong \bR^{2+4n}$, 
\begin{align*}
\cXFP_n^{1} := \Big\{ (z, o, \bs\lambda; \bs s) \in \bR \times \bR \times \bC^n \times \bC^n \;|\; z \neq o, \, z \neq \lambda_i, \, o \neq \lambda_i, \, \lambda_i \neq \lambda_j \textnormal{ for } 1 \leq i \neq i \leq n \Big\} ,
\end{align*}
which are $C^2$ in the $z$-variable, $C^1$ in the $\spec$-variable, and holomorphic in the variables $(\bs\lambda; \bs s)$.
The next result is a generalization of~\cite[Lemma~23]{Dubedat:Double_dimers_conformal_loop_ensembles_and_isomonodromic_deformations}.
It involves Wirtinger derivatives $\frac{\partial}{\partial {\lambda_i}} = \frac{1}{2} \big( \partial_{\Re(\lambda_i)} - \ii \partial_{\Im(\lambda_i)} \big)$ 
and $\frac{\partial}{\partial {s_i}} = \frac{1}{2}  \big( \partial_{\Re(s_i)} - \ii \partial_{\Im(s_i)} \big)$ at the complex variables.

\begin{prop} \label{prop:non_BPZ_PDE}
Suppose that $\PartF \colon \cXFP_n^{1} \to \bC$ belongs to $\FunSpGen(\cXFP_n^{1})$ and solves the PDE
\begin{align} \label{eq:non_BPZ_PDE}
\Bigg[ \frac{\partial^2}{\partial z^2} - \frac{1}{z-\xi} \bigg( \frac{\partial}{\partial z} + \frac{\partial}{\partial \spec} \bigg) 
- \sum_{i=1}^{n} \Bigg( \; & 
\frac{1}{z-\lambda_i} \frac{\partial}{\partial {\lambda_i}} + \frac{s_i}{(z-\lambda_i)^2} \frac{\partial}{\partial {s_i}} 
\\
\; & 
+ \frac{\alpha_i^2}{(z-\lambda_i)^2} +  \frac{2 s_i \alpha_i }{(z - \lambda_i)^3} + \frac{s_i^2 }{(z - \lambda_i)^4}  \Bigg) \Bigg]  \PartF(z, \xi, \bs\lambda; \bs s) = 0 .
\nonumber
\end{align}
Then, with the covariance factor $F(\bs\Lambda_t; \bs S_t)$ from~\eqref{eq:covariance_factor_F},
the process 
\begin{align*}
M_t \coloneqq F(\bs\Lambda_t; \bs S_t) \, \PartF(\Driv_t, \Spec_t, \bs\Lambda_t; \bs S_t) , \qquad t < T_\spec ,
\end{align*}
is a local martingale for the chordal $\SLE_4(-2)$ in $(\bH, 0, \infty)$ with force point $\spec = \Spec_0 \in \bR \setminus \{0\}$.
Moreover, if at the diagonal $\xi=z$ the function $\PartF$ is well-defined and satisfies
\begin{align} \label{eq:der_cancel}
\Big( \tfrac{\partial}{\partial z} + \tfrac{\partial}{\partial \spec} \Big) \PartF(z, \xi, \bs\lambda; \bs s) = 0, 
\end{align}
then $M$ is a local martingale for the chordal $\SLE_4^{\langle\mu\rangle}(-2)$ up to the stopping time $T_{\bs\Lambda}$. 
\end{prop}

\begin{proof}
We have seen in Equation~\eqref{eq:Ito_F} that for the covariance factor $F_t$ from~\eqref{eq:covariance_factor_F},
\begin{align*}
\tfrac{1}{2} \ud \log F_t
= \; & - \sum_{i=1}^{n} \bigg( 
\frac{\alpha_i^2}{(\Driv_t - \Lambda_t^i)^2}
\; + \; \frac{2 \alpha_i S_t^i }{(\Driv_t - \Lambda_t^i)^3} 
\; + \; \frac{ (S_t^i)^2 }{(\Driv_t - \Lambda_t^i)^4}
\bigg) \ud t .
\end{align*}
For ease, we write $\bs X_t \coloneqq (\Driv_t, \Spec_t, \bs\Lambda_t; \bs S_t)$, which is a continuous semimartingale with values in $\cXFP_n^{1}$ for times $t < T_\spec \wedge T_{\bs\Lambda}$ (\ref{eq:continuation_threshold},~\ref{eq:swallow_times}). 
Using It\^o's formula~\eqref{eq:Ito_formula_multi} and~(\ref{eq:Mart_rules1},~\ref{eq:SLE_kappa(rho)_driving.fct}), we have
\begin{align*}
\tfrac{1}{2} \ud \PartF(\bs X_t)
\; = \;\; & \tfrac{1}{2}\big(\tfrac{\partial}{\partial z} \PartF\big)(\bs X_t) \ud \Driv_t 
\; + \; \tfrac{1}{2}\big(\tfrac{\partial}{\partial \spec} \PartF\big)(\bs X_t) \ud \Spec_t 
\; + \; \tfrac{1}{4} \big( \tfrac{\partial^2}{\partial z^2} \PartF\big)(\bs X_t) \underbrace{\ud \langle \Driv,\Driv \rangle_t}_{= \; 4 \ud t} 
\\
\; & + \tfrac{1}{2} \sum_{i=1}^n \bigg(
\big(\tfrac{\partial}{\partial \lambda_i} \PartF\big)(\bs X_t) \ud \Lambda_t^i + \big(\tfrac{\partial}{\partial s_i}\PartF\big)(\bs X_t) \ud S_t^i \bigg)
\\
= \; & \big(\tfrac{\partial}{\partial z} \PartF\big)(\bs X_t) \ud B_t 
+ \Bigg( \big(\tfrac{\partial^2}{\partial z^2} \PartF\big)(\bs X_t) 
- \frac{\big( \big(\tfrac{\partial}{\partial z} \PartF\big)(\bs X_t) + \big(\tfrac{\partial}{\partial \spec} \PartF\big)(\bs X_t) \big)}{\Driv_t - \Spec_t}
\\
\; & \qquad\qquad\qquad\qquad\qquad\qquad   
- \sum_{i=1}^n \bigg( \frac{\big(\tfrac{\partial}{\partial \lambda_i} \PartF\big)(\bs X_t)}{\Driv_t - \Lambda^i_t} 
+ \frac{S_t^{i}\, \big(\tfrac{\partial}{\partial s_i}\PartF\big)(\bs X_t)}{(\Driv_t - \Lambda_t^i)^2} \bigg) \Bigg) \ud t .
\end{align*}
The assumed PDE~\eqref{eq:non_BPZ_PDE} implies that the drift term in $M_t = F_t \, \PartF(\bs X_t)$ vanishes:
\begin{align*}
\ud M_t = \; & 2 F_t \, \big(\tfrac{\partial}{\partial z} \PartF\big)(\bs X_t) \ud B_t .
\end{align*}
After the time $T_\spec$, we can use a version of It\^o's formula for Dirichlet processes (with zero-energy summand of finite $p$-variation for $p\in (1,2)$) from~\cite[Section~2 and Theorem~3.1]{Bertoin:Sur_une_integrale_pour_processus_variation_bornee} 
as in the proof of \Cref{prop:mgle_beyond_ct}.
We again see that thanks to the additional assumption~\eqref{eq:der_cancel}, 
the term involving the local time $\ell^0$ of $B = \frac{1}{2} (\Driv^{\langle \mu\rangle} - \Xi^{\langle \mu\rangle})$ in~\eqref{eq:I^0_def} vanishes, 
and the principal value integral in~\eqref{eq:I^0_def} yields a convergent integral of the form
\begin{align*}
- \int_0^t F_s \, \big(\tfrac{\partial}{\partial z} \PartF + \tfrac{\partial}{\partial \spec} \PartF\big)(\bs X^{\langle \mu\rangle}_s) \ud I^0_s 
\; = \; - 2 \int_0^t F_s \, \frac{\big(\tfrac{\partial}{\partial z} \PartF + \tfrac{\partial}{\partial \spec} \PartF\big)(\bs X^{\langle \mu\rangle}_s)}{Z^{\langle\mu\rangle}_s - \Xi^{\langle\mu\rangle}_s}\ud s.  
\end{align*}
in the It\^o differential $\ud M_t$, 
where $\bs X^{\langle \mu\rangle}_s = (\Driv^{\langle \mu\rangle}_s, \Xi^{\langle \mu\rangle}_s, \bs\Lambda_s; \bs S_s)$ is a Dirichlet process. 
Thus, the drift term in $M_t = F_t \, \PartF(\bs X_t)$ vanishes also for times beyond $T_\spec$, before the time $T_{\bs\Lambda}$. 
\end{proof}

\begin{thm} \label{thm:non_BPZ_PDE2}
Let $\PartF \colon \cXFP_n^{1} \to \bC$ be a locally bounded and Borel measurable function, and $G \colon \cXFP_n^{1} \to \bC$ a continuous function. 
Suppose that the process 
\begin{align*}
M_t \coloneqq \exp \bigg( \int_0^t G(\Driv_r, \Spec_r, \bs\Lambda_r; \bs S_r) \, \ud r \bigg) \, \PartF(\Driv_t, \Spec_t, \bs\Lambda_t; \bs S_t) 
\end{align*}
is a local martingale for the chordal $\SLE_4(-2)$ in $(\bH, 0, \infty)$ with force point $\spec = \Spec_0 \in \bR \setminus \{0\}$.
Then, $\PartF$ is a \textnormal{(}weak\textnormal{)} solution to the PDE
\begin{align} \label{eq:non_BPZ_PDE_cplx}
\Bigg[  \; & 
\frac{\partial^2}{\partial z^2} - \frac{1}{z-\xi} \bigg( \frac{\partial}{\partial z} + \frac{\partial}{\partial \spec} \bigg) 
+ G(z, \spec, \bs\lambda; \bs s) 
\\
\; &  - \sum_{i=1}^{n} \Bigg(
\frac{1}{z-\lambda_i} \frac{\partial}{\partial {\lambda_i}} + \frac{s_i}{(z-\lambda_i)^2} \frac{\partial}{\partial {s_i}} 
+ \frac{1}{z-\lambda_i^*} \frac{\partial}{\partial {\lambda_i^*}} + \frac{s_i^*}{(z-\lambda_i^*)^2} \frac{\partial}{\partial {s_i^*}} \Bigg) \Bigg] \PartF(z, \xi, \bs\lambda; \bs s) = 0 .
\nonumber
\end{align}
\end{thm}

\begin{proof}
First of all, the process $M$ is a complex-valued (generalized) martingale observable in the sense of~\cite[Definition~2.1]{Karrila-Viitasaari:In_prep}, and it satisfies the assumptions for~\cite[Theorem~2.3]{Karrila-Viitasaari:In_prep}.
Second of all, by the identities~(\ref{eq:Mart_rules1},~\ref{eq:SLE_kappa(rho)_driving.fct}), the vector-valued It\^o process
\begin{align} \label{eq:Ito_process}
t \mapsto \big( \Driv_t, \Spec_t, \bs\Lambda_t; \bs S_t \big) 
\end{align}
taking values in $\cXFP_n^{1}$ comprises semimartingales of which the first component satisfies 
$\ud \Driv_t = 2 \ud B_t + \frac{2}{\Spec_t - \Driv_t} \ud t$,
while the other components are finite-variation processes, whose differentials are smooth complex-valued functions (for short times).
It has infinitesimal generator $\cG = \frac{1}{2} \cU_1^2 + \cU_0$, where $\cU_1 \coloneqq 2 \partial_z$ and 
\begin{align*}
\cU_0 \coloneqq \; & - \frac{2}{z - \xi} \big( \partial_z + \partial_\spec \big) 
- \sum_{i=1}^{n} \Bigg( 
\frac{2}{z - \lambda_i} \partial_{\lambda_i} + \frac{2 s_i}{(z - \lambda_i)^2} \partial_{s_i} 
+ \frac{2}{z - \lambda_i^*} \partial_{\lambda_i^*} + \frac{2 s_i^*}{(z - \lambda_i^*)^2} \partial_{s_i^*} 
\Bigg) .
\end{align*}
Moreover, we will check below in \Cref{lem:Hormander} that the (complex) vector fields $\cU_1$ and $\cU_0$ involved in $\cG$ satisfy a version of the H\"ormander bracket condition, 
as phrased in~\cite[Condition~(H) and Remark~3.3]{Karrila-Viitasaari:In_prep}). 
Taking $(z, \xi, \bs\lambda; \bs s) \in \cXFP_n^{1}$, by~\cite[Theorem~2.3]{Karrila-Viitasaari:In_prep} we can then conclude that $\PartF$ is a 
(weak) solution to the following PDE, equivalent to~\eqref{eq:non_BPZ_PDE_cplx}: 
\begin{align} \label{eq:eq:non_BPZ_PDE_diff}
\Big( \cG + G(z, \spec, \bs\lambda; \bs s)  \Big) \PartF(z, \xi, \bs\lambda; \bs s) = 0 .
\end{align} 
Up to verifying \Cref{lem:Hormander}, this concludes the proof thanks to the results in~\cite{Karrila-Viitasaari:In_prep}.
\end{proof}

As a special case, $g_0$ is the identity map, so $g_0'(\lambda_i) = 1$ and $\PreSchwarzian(g_0)(\lambda_i) = 0 = \Schwarzian(g_0)(\lambda_i)$. 
Thus, using~\eqref{eq:Ito_F}, we can write the covariance factor~\eqref{eq:covariance_factor_F} as
\begin{align*}
F(\bs\Lambda_t; \bs S_t) = \; & \exp \bigg( \int_0^t G(\Driv_r, \bs \Lambda_r;  \bs S_r) 
\ud r \bigg) , 
\\
\textnormal{where} \qquad
G(\Driv_t, \bs \Lambda_t;  \bs S_t) 
= \; & - 2 \sum_{i=1}^{n} \bigg( 
\frac{\alpha_i^2}{(\Lambda_t^i - \Driv_t)^2}
\; + \; \frac{ (S_t^i)^2 }{(\Lambda_t^i - \Driv_t)^4}
\; - \; \frac{2 \alpha_i S_t^i }{(\Lambda_t^i - \Driv_t)^3} \bigg) ,
\end{align*}
which is a complex-valued function satisfying the assumptions of \Cref{thm:non_BPZ_PDE2}.
By sending $\spec$ to infinity, we thus recover the confluent BPZ PDE~\eqref{eq:BPZ} related to the local martingale in \Cref{thm:main_mgle}.
This proves \Cref{thm:BPZ_general} stated in the introduction.

\begin{cor}
The function $\tau(\bs\lambda; \bs s) \, \Tr \big( Y^{-1}(\xi, \bs\lambda; \bs s) \, Y(z, \bs\lambda; \bs s) \big)$ solves the PDE~\eqref{eq:non_BPZ_PDE}, 
and the function $\tau(\bs\lambda; \bs s) \, \Tr \big( Y(z, \bs\lambda; \bs s) \big)$ solves the confluent BPZ PDE~\eqref{eq:BPZ}.
\end{cor}
\begin{proof}
One can either check this by direct computation, or invoke \Cref{thm:non_BPZ_PDE2}.
\end{proof}

To finish the proof of \Cref{thm:non_BPZ_PDE2}, we need to verify a version of the H\"ormander bracket condition~\cite{Hormander:The_analysis_of_linear_partial_differential_operators_1},
which is a crucial ingredient --- as showcased by the recent article~\cite{Karrila-Viitasaari:In_prep} (and earlier results of similar nature).
In contrast to earlier literature, we also need to consider complex vector fields. 
Note that the H\"ormander bracket condition does not automatically guarantee smoothness in the complex case.

\begin{lem}\label{lem:Hormander}
Consider the vector fields $\cU_1 \coloneqq 2 \partial_z$ and 
\begin{align*}
\cU_0 \coloneqq \; & - \frac{2}{z - \xi} \big( \partial_z + \partial_\spec \big) 
- \sum_{i=1}^{n} \Bigg( 
\frac{2}{z - \lambda_i} \partial_{\lambda_i} + \frac{2 s_i}{(z - \lambda_i)^2} \partial_{s_i} 
+ \frac{2}{z - \lambda_i^*} \partial_{\lambda_i^*} + \frac{2 s_i^*}{(z - \lambda_i^*)^2} \partial_{s_i^*} 
\Bigg) .
\end{align*}
\begin{itemize}[leftmargin=*]
\item At every $(z, \xi, \bs\lambda; \bs s) \in \cXFP_n^{1} \cap ( \bR^{2} \times (\bC \setminus \bR)^{n} \times \bC^{n})$, 
the Lie algebra generated by $\cU_1$ and the Lie brackets of $\cU_0$ and $\cU_1$ is of full real dimension $4n+2$.

\item At every $(z, \xi, \bs\lambda; \bs s) \in \cXFP_n^{1} \cap ( \bR^{2} \times \bR^{n-k} \times (\bC \setminus \bR)^{k} \times \bC^{n})$, with $0 \leq k \leq n$, 
the Lie algebra generated by $\cU_1$ and the Lie brackets of $\cU_0$ and $\cU_1$ is of full real dimension $3n+k+2$.
\end{itemize}
\end{lem}

\begin{proof}
Consider the commutators $\cU_0^{[\ell]} := \frac{1}{\ell!} \, \big[ \partial_z, \cU_0^{[\ell-1]} \big]$ given by
\begin{align*}
\cU_0^{[\ell]} 
= \frac{2}{(\xi - z)^{\ell+1}} \big( \partial_z + \partial_\spec \big) 
\; + \; \sum_{i=1}^{n} \Bigg( \; & 
\frac{2}{(\lambda_i - z)^{\ell+1}} \, \partial_{\lambda_i}
\; - \; \frac{2 (\ell+1) s_i}{(\lambda_i - z)^{\ell+2}} \, \partial_{s_i} 
\\
\; & 
\, + \, \frac{2}{(\lambda_i^* - z)^{\ell+1}} \, \partial_{\lambda_i^*}
\, - \, \frac{2 (\ell+1) s_i^*}{(\lambda_i^* - z)^{\ell+2}} \, \partial_{s_i^*} \Bigg) , \quad \ell \geq 1 .
\end{align*}
We can write 
\begin{align}\label{eq:vector_fields}
\big( \cU_0^{[1]}, \ldots, \cU_0^{[4n+1]} \big)^t 
=  \frac{2}{\xi-z} \, \Delta \, \Big( \; &\partial_z +  \partial_\spec , \partial_{\lambda_1} , \ldots , \partial_{\lambda_n} , \partial_{\lambda_1^*} , \ldots , \partial_{\lambda_n^*} , 
 \\
\; & - \tfrac{ \ell \, s_{1}}{\lambda_{1}-z} \,\partial_{s_1} , \ldots ,  - \tfrac{ \ell \, s_{n}}{\lambda_{n}-z}\,\partial_{s_n}, - \tfrac{ \ell \, s_{1}^*}{\lambda_{1}^*-z} \,\partial_{s_1^*} , \ldots , - \tfrac{ \ell \, s_{n}^*}{\lambda_{n}^*-z}\,\partial_{s_n^*} \Big)^t ,
\nonumber
\end{align}
where $\Delta = (\Delta_{i,j})_{i,j=1}^{4n+1}$ is a Vandermonde type matrix, with generically non-zero determinant:
\begin{align*}
\Delta_{\ell,1} = \; & \frac{1}{(\xi-z)^\ell} , && \ell = 1,2,\ldots, 4n+1 , \\
\Delta_{\ell,i} = \; & \frac{1}{ (\lambda_{i-1}-z)^{\ell}} , && i = 2,\ldots,n+1, \; \ell = 1,2,\ldots, 4n+1 , \\
\Delta_{\ell,i} = \; & \frac{1}{ (\lambda_{i-1}^*-z)^{\ell}} , && i = n+2,\ldots,2n+1, \; \ell = 1,2,\ldots, 4n+1 , \\
\Delta_{\ell,i} = \; & \frac{1}{(\lambda_{i-n-1}-z)^{\ell}} , && i = 2n+2,\ldots,3n+1, \; \ell = 1,2,\ldots, 4n+1 , \\
\Delta_{\ell,i} = \; & \frac{1}{(\lambda_{i-n-1}^*-z)^{\ell}} , && i = 3n+2,\ldots,4n+1, \; \ell = 1,2,\ldots, 4n+1 .
\end{align*}
Thus, at every $(z, \xi, \bs\lambda; \bs s) \in \cXFP_n \cap ( \bR^{2} \times (\bC \setminus \bR)^{n} \times \bC^{n})$, 
we can solve for the vector fields on the right side of~\eqref{eq:vector_fields} in terms of the commutators on the left side.
Because we also have $\cU_1 = 2 \partial_z$ by definition, we see that the Lie algebra generated by $\cU_0$, $\cU_1$ and their Lie brackets is of real dimension $4n+2$. 
In turn, at points $(z, \xi, \bs\lambda; \bs s) \in \cXFP_n \cap ( \bR^{2} \times \bR^{n-k} \times (\bC \setminus \bR)^{k} \times \bC^{n})$ where $k$ of the variables $\lambda_j$ are real, 
the Vandermonde type matrix $V$ has identical columns, but also, those derivatives $\partial_{\lambda_i^*}$ are absent, and we get the real dimension $3n+k+2$. 
\end{proof}

\subsection{Martingales with general pole structure\label{subsec:mgle_higher_poles}}

Let us now explain how to generalize the above martingale construction (\Cref{thm:SLE_4(-2)_mart}) to higher order singularities as in the IDEs~(\ref{eq:non-Fuchsian-A},~\ref{eq:LaxUigen},~\ref{eq:LaxVigen}). 
We begin with the expressions
\begin{align*}
\ud \widetilde Y_t^{-1} 
= \; & -\widetilde Y_t^{-1} \bigg( \widetilde A_t \ud \Spec_t 
\; + \; \sum_{i=1}^n U_i(\Spec_t, \bs \Lambda_t; \bs S_t) \ud \Lambda_t^i 
\; + \; \sum_{i=1}^n \sum_{k=1}^{r_i} V_{i,k}(\Spec_t, \bs \Lambda_t; \bs S_t) \ud S^{i,k}_t \bigg) ,
\\
\ud Y_t = \; & \Big( 2 A_t \ud B_t - \frac{2 A_t}{\Driv_t - \Spec_t} \ud t + \Tr (A_t^2) \ud t \Big) Y_t ,
\end{align*}
which, using the Loewner evolutions~\eqref{SLEeq-irrgen} of the Birkhoff invariants from \Cref{thm:Sto-Bir-Gen}
and the standard Loewner evolutions of the other processes, and the identity~\eqref{eq:LaxVigen}, give us 
\begin{align*}
\ud \big(\widetilde Y_t^{-1} Y_t\big) 
= \; & 2\widetilde Y_t^{-1} A_t Y_t \ud B_t \; + \; \Tr (A_t^2) \, \widetilde Y_t^{-1} Y_t \ud t \\
\, & - 2\widetilde Y_t^{-1} \sum_{i=1}^n \sum_{k=0}^{r_i} \Bigg( 
\frac{A_{i,k}(\bs\Lambda_t; \bs S_t)}{(\Spec_t - \Lambda^i_t)^{k+1} (\Spec_t - \Driv_t)} 
\; + \; \frac{A_{i,k}(\bs\Lambda_t; \bs S_t)}{(\Spec_t - \Lambda^i_t)^{k+1} (\Driv_t - \Lambda^i_t )} \\
\; & \; + \;  \frac{\one\{k \geq 1\}}{(\Driv_t - \Lambda_t^i)^2} \frac{ A_{i,k}(\bs\Lambda_t; \bs S_t) }{ (\Spec_t - \Lambda_t^i)^{k}} 
\; + \; \frac{A_{i,k}(\bs\Lambda_t; \bs S_t)}{(\Driv_t - \Lambda^i_t)^{k+1} (\Driv_t - \Spec_t)}\Bigg) Y_t \ud t
\\
= \; & 2\widetilde Y_t^{-1} A_t Y_t \ud B_t \; + \; \Tr (A_t^2) \widetilde Y_t^{-1} Y_t \ud t \\
\, & - 2\widetilde Y_t^{-1} \sum_{i=1}^n \sum_{k=1}^{r_i} A_{i,k}(\bs\Lambda_t; \bs S_t) \;
\frac{\big( (\Spec_t - \Lambda^i_t)^k + (\Spec_t - \Lambda^i_t) (\Driv_t - \Lambda^i_t)^{k-1} \big)}{(\Spec_t - \Lambda^i_t)^{k} (\Spec_t - \Driv_t) (\Driv_t - \Lambda_t^i)^{k+3}} 
\; Y_t \ud t .
\end{align*}
Next, one should write $\Tr (A_t^2)$ in a form analogous to~\eqref{eq:TrL2_again} (by an analogue of \Cref{lem:Hamiltonians}):
\begin{align*}
\Tr (A_t^2) 
= \; & \sum_{i=1}^{n} \sum_{k\geq 1} (-1)^k \, \frac{\ell_{i,k}(\bs\Lambda_t; \bs S_t)}{(\Driv_t - \Lambda_t^i)^{k}} ,
\end{align*} 
and compute the variation of the tau-function analogous to~\eqref{eq:Ito_tau}:
\begin{align*}
\ud \log \tau_t 
= \; & \sum_{i=1}^n H_{\Lambda^i_t}(\bs \Lambda_t; \bs S_t) \ud \Lambda^i_t 
\; + \; \sum_{i=1}^n \sum_{k=1}^{r_i} H_{S^{i,k}_t}(\bs \Lambda_t; \bs S_t) \ud S^{i,k}_t .
\end{align*}
As before, thanks to the explicit time-evolution~\eqref{SLEeq-irrgen} of the Birkhoff invariants from \Cref{thm:Sto-Bir-Gen},
some of the terms arising from the Hamiltonians would then cancel out with some of the terms arising from $\Tr (A_t^2)$. 
The other terms would then be canceled with the time-evolution of $\ud \log F_t$ 
for a judiciously chosen covariance factor $F_t$, generalizing~\eqref{eq:covariance_factor_F}, which we expect to involve higher Schwarzians in the case of irregular singularities of higher Poincar\'e rank.
We would then expect that, as before, we have
\begin{align*}
\ud \big(\widetilde Y_t^{-1} Y_t\big) 
= \; & 2\widetilde Y_t^{-1} A_t Y_t \ud B_t - (\ud \log F_t + \ud \log \tau_t) \, \widetilde Y_t^{-1} Y_t \ud t .
\end{align*}
We leave the computation of $\ud \log \tau_t$ and the determination of $F_t$ to future work.

%% file: jabref_database-2025.bib
@Article{Bauer-Bernard:Conformal_transformations_and_SLE_partition_function_martingale,
  Title                    = {Conformal transformations and the {$\mathrm{SLE}$} partition function martingale},
  Author                   = {Michel Bauer and Denis Bernard},
  Journal                  = {Ann. Henri Poincar\'e},
  Year                     = {2004},
  Number                   = {2},
  Pages                    = {289-326},
  Volume                   = {5}
}

@Article{Bauer-Bernard:Conformal_field_theories_of_SLEs,
  Title                    = {Conformal field theories of stochastic {L}oewner evolutions},
  Author                   = {Michel Bauer and Denis Bernard},
  Journal                  = {Comm. Math. Phys.},
  Year                     = {2003},
  Number                   = {3},
  Pages                    = {493-521},
  Volume                   = {239}
}

@Article{Bauer-Bernard:SLE_martingales_and_Virasoro_algebra,
  Title                    = {{$\mathrm{SLE}$} martingales and the {V}irasoro algebra},
  Author                   = {Michel Bauer and Denis Bernard},
  Journal                  = {Phys. Lett. B},
  Year                     = {2003},
  Number                   = {3-4},
  Pages                    = {309-316},
  Volume                   = {557}
}

@Article{Bauer-Bernard:SLE_growth_processes_and_CFTs,
  Title                    = {$\mathrm{SLE}_\kappa$ growth processes and conformal field theories},
  Author                   = {Michel Bauer and Denis Bernard},
  Journal                  = {Phys. Lett. B},
  Year                     = {2002},
  Number                   = {1-2},
  Pages                    = {135-138},
  Volume                   = {543}
}

@Article{BPZ:Infinite_conformal_symmetry_in_2D_QFT,
  author    = {Alexander A. Belavin and Alexander M. Polyakov and Alexander B. Zamolodchikov},
  journal   = {Nucl. Phys. B},
  title     = {Infinite conformal symmetry in two-dimensional quantum field theory},
  year      = {1984},
  number    = {2},
  pages     = {333--380},
  volume    = {241},
}

@Article{Dubedat:SLE_and_Virasoro_representations_localization,
  author    = {Julien Dub{\'e}dat},
  title     = {$\mathrm{SLE}$ and {V}irasoro representations: localization},
  journal   = {Comm. Math. Phys.},
  year      = {2015},
  volume    = {336},
  number    = {2},
  pages     = {695-760},
  date      = {1},
}

@Book{Durrett:Stochastic_calculus,
  Title                    = {Stochastic calculus: a practical introduction},
  Author                   = {Richard Durrett},
  Publisher                = {CRC Press LLC},
  Year                     = {1996},
  Series                   = {Probability and Stochastics Series},

  Comment                  = {(ISBN: 0-8493-8071-5)}
}

@Unpublished{Friedrich:On_connections_of_CFT_and_SLE,
  author    = {Roland Friedrich},
  note      = {Preprint in arXiv:math-ph/0410029},
  title     = {On connections of conformal field theory and stochastic {L}oewner evolution},
  year      = {2004},
}

@Article{Friedrich-Kalkkinen:On_CFT_and_SLE,
  Title                    = {On conformal field theory and stochastic {L}oewner evolution},
  Author                   = {Roland Friedrich and Jussi Kalkkinen},
  Journal                  = {Nucl. Phys. B},
  Year                     = {2004},
  Number                   = {3},
  Pages                    = {279-302},
  Volume                   = {687}
}

@Article{Friedrich-Werner:Conformal_restriction_highest_weight_representations_and_SLE,
  Title                    = {Conformal restriction, highest weight representations and {$\mathrm{SLE}$}},
  Author                   = {Roland Friedrich and Wendelin Werner},
  Journal                  = {Comm. Math. Phys.},
  Year                     = {2003},
  Number                   = {1},
  Pages                    = {105-122},
  Volume                   = {243}
}

@Book{Kemppainen:SLE_book,
  author    = {Antti Kemppainen},
  publisher = {Springer Cham},
  title     = {Schramm-{L}oewner evolution},
  year      = {2017},
  series    = {SpringerBriefs in Mathematical Physics},
  volume    = {24},
  comment   = {[http://www.helsinki.fi/sle-book]},
}

@Article{Kenyon:Conformal_invariance_of_loops_in_the_double-dimer_model,
  Title                    = {Conformal invariance of loops in the double-dimer model},
  Author                   = {Richard W. Kenyon},
  Journal                  = {Comm. Math. Phys.},
  Year                     = {2014},
  Number                   = {2},
  Pages                    = {477-497},
  Volume                   = {326}
}

@Article{Kenyon-Wilson:Double_dimer_pairings_and_skew_Young_diagrams,
  author    = {Richard W. Kenyon and David B. Wilson},
  journal   = {Electron. J. Combin.},
  title     = {Double-dimer pairings and skew {Y}oung diagrams},
  year      = {2011},
  number    = {1},
  pages     = {1--22},
  volume    = {18},
}

@InProceedings{Kontsevich:CFT_SLE_and_phase_boundaries,
  author    = {Maxim Kontsevich},
  booktitle = {Oberwolfach Arbeitstagung},
  title     = {{CFT}, $\mathrm{SLE}$, and phase boundaries},
  year      = {2003},
}

@Article{Kontsevich-Suhov:On_Malliavin_measures_SLE_and_CFT,
  author    = {Maxim Kontsevich and Yuri Suhov},
  title     = {On {M}alliavin measures, $\mathrm{SLE}$, and {CFT}},
  journal   = {P. Steklov I. Math.},
  year      = {2007},
  volume    = {258},
  number    = {1},
  pages     = {100--146},
  comment   = {Tr. Mat. Inst. Steklova},
}

@Article{Loewner:Untersuchungen_uber_schlichte_konforme_Abbildungen_des_Einheitskreises,
  author    = {Charles Loewner},
  journal   = {Math. Ann.},
  title     = {Untersuchungen {\"u}ber schlichte konforme {A}bbildungen des {E}inheitskreises {I}},
  year      = {1923},
  pages     = {103--121},
  volume    = {89},
}

@Article{Miller-Sheffield:Imaginary_geometry1,
  author    = {Jason Miller and Scott Sheffield},
  journal   = {Probab. Theory Related Fields},
  title     = {Imaginary geometry {I}: interacting $\mathrm{SLE}$s},
  year      = {2016},
  number    = {3-4},
  pages     = {553--705},
  volume    = {164},
}

@Article{Rohde-Schramm:Basic_properties_of_SLE,
  author    = {Steffen Rohde and Oded Schramm},
  journal   = {Ann. of Math.},
  title     = {Basic properties of $\mathrm{SLE}$},
  year      = {2005},
  number    = {2},
  pages     = {883--924},
  volume    = {161},
}

@Article{Schramm:Scaling_limits_of_LERW_and_UST,
  author    = {Oded Schramm},
  journal   = {Israel J. Math.},
  title     = {Scaling limits of loop-erased random walks and uniform spanning trees},
  year      = {2000},
  number    = {1},
  pages     = {221--288},
  volume    = {118},
}

@Article{Schramm-Wilson:SLE_coordinate_changes,
  author    = {Oded Schramm and David B. Wilson},
  journal   = {New York J. Math.},
  title     = {$\mathrm{SLE}$ coordinate changes},
  year      = {2005},
  pages     = {659--669},
  volume    = {11},
}

@Article{Sheffield:Exploration_trees_and_CLEs,
  author    = {Scott Sheffield},
  journal   = {Duke Math. J.},
  title     = {Exploration trees and conformal loop ensembles},
  year      = {2009},
  number    = {1},
  pages     = {79--129},
  volume    = {147},
}

@Article{Sheffield-Werner:CLEs,
  author    = {Scott Sheffield and Wendelin Werner},
  journal   = {Ann. of Math.},
  title     = {Conformal loop ensembles: {T}he {M}arkovian characterization and the loop-soup construction},
  year      = {2012},
  number    = {3},
  pages     = {1827--1917},
  volume    = {176},
}

@Article{Werner-Wu:On_conformally_invariant_CLE_explorations,
  author    = {Wendelin Werner and Hao Wu},
  journal   = {Comm. Math. Phys.},
  title     = {On conformally invariant {CLE} explorations},
  year      = {2013},
  number    = {3},
  pages     = {637--661},
  volume    = {320},
}

@PhdThesis{Zhan:Thesis,
  Title                    = {Random {L}oewner chains in {R}iemann surfaces},
  Author                   = {Dapeng Zhan},
  School                   = {California Institute of Technology},
  Year                     = {2004}
}

@Article{Peltola-Wu:Global_and_local_multiple_SLEs_and_connection_probabilities_for_level_lines_of_GFF,
  author  = {Eveliina Peltola and Hao Wu},
  journal = {Comm. Math. Phys.},
  title   = {Global and local multiple $\mathrm{SLE}$s for $\kappa \leq 4$ and connection probabilities for level lines of {GFF}},
  year    = {2019},
  number  = {2},
  pages   = {469--536},
  volume  = {366},
}

@InProceedings{Schramm:ICM,
  author    = {Oded Schramm},
  booktitle = {Proceedings of the ICM 2006, Madrid, Spain},
  title     = {Conformally invariant scaling limits, an overview and a collection of problems},
  year      = {2006},
  pages     = {513--543},
  publisher = {European Mathematical Society},
  volume    = {1},
}

@Article{Peltola:Towards_CFT_for_SLEs,
  author  = {Eveliina Peltola},
  journal = {J. Math. Phys.},
  title   = {Towards a conformal field theory for {S}chramm-{L}oewner evolutions},
  year    = {2019},
  note    = {Special issue (Proc. ICMP, Montreal, July 2018)},
  number  = {10},
  pages   = {103305},
  volume  = {60},
}

@Article{Dubedat:Double_dimers_conformal_loop_ensembles_and_isomonodromic_deformations,
  author  = {Julien Dub{\'e}dat},
  journal = {J. Eur. Math. Soc.},
  title   = {Double dimers, conformal loop ensembles, and isomonodromic deformations},
  year    = {2019},
  number  = {1},
  pages   = {1--54},
  volume  = {21},
}

@Article{Basok-Chelkak:Tau_functions_a_la_Dubedat_and_probabilities_of_cylindrical_events_for_double_dimers_and_CLE4,
  author  = {Mikhail Basok and Dmitry Chelkak},
  journal = {J. Eur. Math. Soc.},
  title   = {Tau-functions {\`a} la {D}ub{\'e}dat and probabilities of cylindrical events for double-dimers and {$\mathrm{CLE}(4)$}},
  year    = {2021},
  number  = {8},
  pages   = {2787--2832},
  volume  = {23},
}

@Article{Fock-Goncharov:Moduli_spaces_of_local_systems_and_higher_Teichmuller_theory,
  author  = {Vladimir V. Fock and Alexander B. Goncharov},
  journal = {Publ. Math. Inst. Hautes \'Etudes Sci.},
  title   = {Moduli spaces of local systems and higher {T}eichm\"uller theory},
  year    = {2006},
  pages   = {1--211},
  volume  = {103},
}

@Book{Hormander:The_analysis_of_linear_partial_differential_operators_1,
  author    = {Lars H{\"o}rmander},
  publisher = {Springer-Verlag},
  title     = {The analysis of linear partial differential operators {I}: {D}istribution theory and {F}ourier analysis},
  year      = {1990},
  address   = {Berlin Heidelberg},
  edition   = {2},
  series    = {Grundlehren der mathematischen Wissenschaften},
  volume    = {256},
}

@Unpublished{FLPW:Multiple_SLEs_Coulomb_gas_integrals_and_pure_partition_functions,
  author = {Yu Feng and Mingchang Liu and Eveliina Peltola and Hao Wu},
  note   = {Preprint in arXiv:2406.06522},
  title  = {Multiple {SLE}s for {$\kappa\in (0,8)$:} {C}oulomb gas integrals and pure partition functions},
  year   = {2024},
}

@Article{GIL:Conformal_field_theory_of_PVI,
  author  = {Oleksandr Gamayun and Nikolai~Z. Iorgov and Oleg Lisovyy},
  journal = {JHEP},
  title   = {Conformal field theory of {P}ainlev{\'e} {VI}},
  year    = {2012},
  number  = {38},
  pages   = {1--24},
  volume  = {10},
}

@Article{ILT:Isomonodromic_tau-functions_from_Liouville_conformal_blocks,
  author  = {Nikolai Iorgov and Oleg Lisovyy and J{\"o}rg Teschner},
  journal = {Comm. Math. Phys.},
  title   = {Isomonodromic tau-functions from {L}iouville conformal blocks},
  year    = {2015},
  number  = {2},
  pages   = {671--694},
  volume  = {336},
}

@Article{Bai-Wan:On_the_crossing_estimates_of_simple_conformal_loops_ensembles,
  author  = {Tianyi Bai and Yijun Wan},
  journal = {Int. Math. Res. Not.},
  title   = {On the crossing estimates of simple conformal loops ensembles},
  year    = {2022},
  number  = {13},
  pages   = {11645--11683},
  volume  = {2023},
}

@Unpublished{Karrila-Peltola:Boundary_double-dimer_patterns_and_CFT,
  author = {Alex Karrila and Eveliina Peltola},
  note   = {In preparation},
  title  = {Boundary double-dimer patterns and conformal field theory},
  year   = {2026},
}

@Article{JMU:Monodromy_preserving_deformation_of_linear_ODEs_with_rational_coefficients1,
  author  = {Michio Jimbo and Tetsuji Miwa and Kimio Ueno},
  journal = {Phys. D},
  title   = {Monodromy preserving deformation of linear ordinary differential equations with rational coefficients. {I}. General theory and {$\tau$}-function},
  year    = {1981},
  number  = {2},
  pages   = {306--352},
  volume  = {2},
}

@Article{ILT:Painleve_VI_connection_problem_and_monodromy_of_c_1_conformal_blocks,
  author  = {Nikolai Iorgov and Oleg Lisovyy and Yu Tykhyy},
  journal = {JHEP},
  title   = {Painlev{\'e} {VI} connection problem and monodromy of {$c=1$} conformal blocks},
  year    = {2013},
  number  = {29},
  pages   = {1--26},
  volume  = {12},
}

@Book{BBT:Introduction_to_classical_integrable_systems,
  author    = {Olivier Babelon and Denis Bernard and Michel Talon},
  publisher = {Cambridge University Press},
  title     = {Introduction to classical integrable systems},
  year      = {2003},
}

@Book{FIKN:Painleve_transcendents_Riemann-Hilbert_approach,
  author    = {Athanassios S. Fokas and Alexander R. Its and Andrei A. Kapaev and Victor Yu Novokshenov},
  publisher = {American Mathematical Society},
  title     = {Painlev{\'e} transcendents: the {R}iemann-{H}ilbert approach},
  year      = {2023},
  series    = {Mathematical Surveys and Monographs},
  volume    = {128},
}

@Article{Jimbo-Miwa:Monodromy_preserving_deformation_of_linear_ODEs_with_rational_coefficients2,
  author  = {Michio Jimbo and Tetsuji Miwa},
  journal = {Phys. D},
  title   = {Monodromy preserving deformation of linear ordinary differential equations with rational coefficients. {II}.},
  year    = {1981},
  number  = {3},
  pages   = {407--448},
  volume  = {2},
}

@Book{Balogh-Harnad:Tau_functions_and_their_applications,
  author    = {John Harnad and Ferenc Balogh},
  publisher = {Cambridge University Press},
  title     = {Tau functions and their applications},
  year      = {2021},
  series    = {Monographs on Mathematical Physics},
}

@Article{Gaiotto-Teschner:Irregular_singularities_in_Liouville_theory_and_Argyres-Douglas_type_gauge_theories,
  author  = {Davide Gaiotto and J{\"o}rg Teschner},
  journal = {JHEP},
  title   = {Irregular singularities in {L}iouville theory and {A}rgyres-{D}ouglas type gauge theories},
  year    = {2012},
  pages   = {1--79},
  volume  = {12},
}

@Article{BEH:Partition_functions_for_matrix_models_and_isomonodromic_tau_functions,
  author  = {Marco Bertola and Bertrand Eynard and John Harnad},
  journal = {J. Phys. A},
  title   = {Partition functions for matrix models and isomonodromic tau functions},
  year    = {2003},
  number  = {12},
  pages   = {3067--3084},
  volume  = {36},
}

@Article{SMJ:Holonomic_quantum_fields3,
  author  = {{Ken-iti} Sato and Tetsuji Miwa and Michio Jimbo},
  journal = {Publ. Res. Inst. Math. Sci. Kyoto},
  title   = {Holonomic quantum fields. 3.},
  year    = {1979},
  number  = {2},
  pages   = {577--629},
  volume  = {15},
}

@Article{Okamoto:On_the_tau-function_of_Painleve_equations,
  author  = {Kazuo Okamoto},
  journal = {Phys. D},
  title   = {On the {$\tau$}-function of the {P}ainlev\'e {VI} equations},
  year    = {1981},
  number  = {3},
  pages   = {525--535},
  volume  = {2},
}

@Article{Schlesinger:Uber_eine_Klasse_von_Differentialsystemen_beliebiger_Ordnung_mit_festen_kritischen_Punkte,
  author  = {Ludwig Schlesinger},
  journal = {J. Reine Angew. Math.},
  title   = {{\"U}ber eine {K}lasse von {D}ifferentialsystemen beliebiger {O}rdnung mit festen kritischen {P}unkten},
  year    = {1912},
  pages   = {96--145},
  volume  = {141},
}

@Article{Birkhoff:The_generalized_riemann_problem_for_linear_differential_equations_and_the_allied_problems_for_linear_difference_and_q-difference_equations,
  author  = {George D. Birkhoff},
  journal = {Proc. Amer. Acad. Arts Sci.},
  title   = {The generalized riemann problem for linear differential equations and the allied problems for linear difference and {$q$}-difference equations},
  year    = {1913},
  pages   = {521--568},
  volume  = {49},
}

@Article{Malmquist:Sur_les_equations_differentielles_du_second_ordre_dont_lintegrale_generale_a_ses_points_critiques_fixes,
  author  = {Johannes Malmquist},
  journal = {Ark. Mat.},
  title   = {Sur les \'equations diff\'erentielles du second ordre dont l'int\'egrale g\'en\'erale a ses points critiques fixes},
  year    = {1922},
  pages   = {1--89},
  volume  = {17},
}

@Article{Painleve:Sur_les_equations_differentielles_du_second_ordre_a_points_critiques_fixes,
  author  = {Paul Painlev\'e},
  journal = {C. R. Acad. Sci. Paris S{\'e}r. I Math.},
  title   = {Sur les \'equations diff\'erentielles du second ordre \`a points critiques fixes},
  year    = {1906},
  pages   = {1111--1117},
  volume  = {143},
}

@Article{Gavrylenko-Lisovyy:Fredholm_determinant_and_Nekrasov_sum_representations_of_isomonodromic_tau_functions,
  author  = {Pavlo Gavrylenko and Oleg Lisovyy},
  journal = {Comm. Math. Phys.},
  title   = {Fredholm determinant and {N}ekrasov sum representations of isomonodromic tau functions},
  year    = {2018},
  number  = {1},
  pages   = {1--58},
  volume  = {363},
}

@Article{GIL:Higher_rank_isomonodromic_deformations_and_W-algebras,
  author  = {Pavlo Gavrylenko and Nikolai Iorgov and Oleg Lisovyy},
  journal = {Lett. Math. Phys.},
  title   = {Higher rank isomonodromic deformations and {W}-algebras},
  year    = {2019},
  number  = {2},
  pages   = {327--364},
  volume  = {110},
}

@Article{Desiraju:Fredholm_determinant_representation_of_the_homogeneous_Painleve2,
  author  = {Harini Desiraju},
  journal = {Nonlinearity},
  title   = {Fredholm determinant representation of the homogeneous {P}ainlev{\'e} {II} {$\tau$}-function},
  year    = {2021},
  number  = {9},
  pages   = {6507--6538},
  volume  = {34},
}

@Article{DMDG:Isomonodromic_tau_functions_on_a_torus_as_Fredholm_determinants_and_charged_partitions,
  author  = {Fabrizio Del~Monte and Harini Desiraju and Pavlo Gavrylenko},
  journal = {Comm. Math. Phys.},
  title   = {Isomonodromic tau functions on a torus as {F}redholm determinants, and charged partitions},
  year    = {2023},
  number  = {3},
  pages   = {1029--1084},
  volume  = {398},
}

@Article{BLMST:On_Painleve_gauge_theory_correspondence,
  author  = {Giulio Bonelli and Oleg Lisovyy and Kazunobu Maruyoshi and Antonio Sciarappa and Alessandro Tanzini},
  journal = {Lett. Math. Phys.},
  title   = {On {P}ainlev\'e/gauge theory correspondence},
  year    = {2017},
  number  = {12},
  pages   = {2359--2413},
  volume  = {107},
}

@Article{BDMGT:N_2_Gauge_theory_free_fermions_on_the_torus_and_Painleve_VI,
  author  = {Giulio Bonelli and Fabrizio Del~Monte and Pavlo Gavrylenko and Alessandro Tanzini},
  journal = {Comm. Math. Phys.},
  title   = {{$\mathcal {N}$ = $2^*$} {G}auge theory, free fermions on the torus and {P}ainlev\'e {VI}},
  year    = {2020},
  number  = {2},
  pages   = {1381--1419},
  volume  = {377},
}

@Article{Desiraju:Painleve_CFT_correspondence_on_a_torus,
  author  = {Harini Desiraju},
  journal = {J. Math. Phys.},
  title   = {Painlev{\'e}/{CFT} correspondence on a torus},
  year    = {2022},
  number  = {8},
  pages   = {081102-1--16},
  volume  = {63},
}

@Article{GIL:How_instanton_combinatorics_solves_Painleve_VI_V_and_IIIs,
  author  = {Oleksandr Gamayun and Nikolai~Z. Iorgov and Oleg Lisovyy},
  journal = {J. Phys. A},
  title   = {How instanton combinatorics solves {P}ainlev\'e {VI}, {V} and {IIIs}},
  year    = {2013},
  number  = {33},
  pages   = {1--29},
  volume  = {46},
}

@Article{ILP:Monodromy_dependence_and_connection_formulae_for_isomonodromic_tau_functions,
  author  = {Alexander R. Its and Oleg Lisovyy and Andrei Prokhorov},
  journal = {Duke Math. J.},
  title   = {Monodromy dependence and connection formulae for isomonodromic tau functions},
  year    = {2018},
  number  = {7},
  pages   = {1347--1432},
  volume  = {167},
}

@Article{ILT:Connection_problem_for_the_Sine-Gordon_Painleve_III_tau_function_and_irregular_conformal_blocks,
  author  = {Alexander R. Its and Oleg Lisovyy and Yuriy Tykhyy},
  journal = {Int. Math. Res. Not.},
  title   = {Connection problem for the {S}ine-{G}ordon/{P}ainlev\'e {III} tau function and irregular conformal blocks},
  year    = {2015},
  number  = {18},
  pages   = {8903--8924},
  volume  = {2015},
}

@Article{Its-Prokhorov:On_some_Hamiltonian_properties_of_the_isomonodromic_tau_functions,
  author  = {Alexander R. Its and Andrei Prokhorov},
  journal = {Rev. Math. Phys.},
  title   = {On some {H}amiltonian properties of the isomonodromic tau functions},
  year    = {2018},
  number  = {7},
  pages   = {1840008},
  volume  = {30},
}

@Unpublished{DMDG:Modular_transformations_of_tau_functions_and_conformal_blocks_on_the_torus,
  author = {Fabrizio Del~Monte and Harini Desiraju and Pavlo Gavrylenko},
  note   = {Preprint in arxiv:2508.14030},
  title  = {Modular transformations of tau functions and conformal blocks on the torus},
  year   = {2025},
}

@Article{DMDG:Monodromy_dependence_and_symplectic_geometry_of_isomonodromic_tau_functions_on_the_torus,
  author  = {Fabrizio Del~Monte and Harini Desiraju and Pavlo Gavrylenko},
  journal = {J. Phys. A},
  title   = {Monodromy dependence and symplectic geometry of isomonodromic tau functions on the torus},
  year    = {2023},
  number  = {29},
  pages   = {1--24},
  volume  = {56},
}

@Article{GMR:Isomonodromic_deformations_confluence_reduction_and_quantisation,
  author  = {Ilia Gaiur and Marta Mazzocco and Vladimir Rubtsov},
  journal = {Comm. Math. Phys.},
  title   = {Isomonodromic deformations: {C}onfluence, reduction, and quantisation},
  year    = {2023},
  number  = {2},
  pages   = {1385--1461},
  volume  = {400},
}

@Article{BHH:Hamiltonian_structure_of_rational_isomonodromic_deformation_systems,
  author  = {Marco Bertola and John Harnad and Jacques Hurtubise},
  journal = {J. Math. Phys.},
  title   = {Hamiltonian structure of rational isomonodromic deformation systems},
  year    = {2023},
  number  = {8},
  pages   = {083502},
  volume  = {64},
}

@Article{BJL:Birkhoff_invariants_and_Stokes_multipliers_for_meromorphic_linear_differential_equations,
  author  = {Werner Balser and Wolfgang B. Jurkat and Donald A. Lutz},
  journal = {J. Math. Anal. Appl.},
  title   = {Birkhoff invariants and {S}tokes{'} multipliers for meromorphic linear differential equations},
  year    = {1979},
  number  = {1},
  pages   = {48--94},
  volume  = {71},
}

@Article{Lehmkuehler:Trunks_of_CLE4_explorations,
  author  = {Matthis Lehmkuehler},
  journal = {Ann. Appl. Probab.},
  title   = {The trunks of {CLE(4)} explorations},
  year    = {2023},
  number  = {5},
  pages   = {3387--3417},
  volume  = {33},
}

@Article{MSW:CLE_percolations,
  author  = {Jason Miller and Scott Sheffield and Wendelin Werner},
  journal = {Forum Math. Pi},
  title   = {{CLE} percolations},
  year    = {2017},
  number  = {e4},
  pages   = {1--102},
  volume  = {5},
}

@Article{Bertoin:Sur_une_integrale_pour_processus_variation_bornee,
  author  = {Jean Bertoin},
  journal = {Ann. Probab.},
  title   = {Sur une int\'egrale pour les processus \`a $\alpha$-variation born\'ee},
  year    = {1989},
  number  = {4},
  pages   = {1521--1535},
  volume  = {17},
}

@Unpublished{Karrila-Viitasaari:In_prep,
  author = {Alex Karrila and Lauri Viitasaari},
  note   = {Preprint in arXiv:2601.10539},
  title  = {Smoothness of martingale observables and generalized {F}eynman-{K}ac formulas},
  year   = {2026},
}

@Book{Fukushima:Dirichlet_forms_and_Markov_processes,
  author    = {Masatoshi Fukushima},
  publisher = {North-Holland},
  title     = {Dirichlet forms and {M}arkov processes},
  year      = {1980},
  address   = {Amsterdam-New York; Kodansha, Ltd., Tokyo},
  volume    = {23},
}

@InCollection{Follmer:Calcul_dIto_sans_probabilites,
  author    = {Hans F{\"o}llmer},
  booktitle = {Seminar on Probability XV (Univ. Strasbourg)},
  publisher = {Springer},
  title     = {Calcul {d'It\^o} sans probabilit{\'e}s},
  year      = {1981},
  address   = {Berlin},
  pages     = {143--150},
  series    = {Lecture Notes in Mathematics},
  volume    = {850},
}

@Article{CPZ:Interactions_of_irregular_Gaiotto_states_in_Liouville_theory,
  author  = {Sang-Kwan Choi and Dimitri Polyakov and Cong Zhang},
  journal = {Eur. Phys. J. C},
  title   = {Interactions of irregular {G}aiotto states in {L}iouville theory},
  year    = {2018},
  number  = {507},
  pages   = {1--17},
  volume  = {78},
}

@Article{HLR:Flat_connections_from_irregular_conformal_blocks,
  author  = {Babak Haghighat and Yihua Liu and Nicolai Reshetikhin},
  journal = {Comm. Math. Phys.},
  title   = {Flat connections from irregular conformal blocks},
  year    = {2025},
  number  = {138},
  pages   = {1--30},
  volume  = {406},
}

@Book{Ablowitz-Fokas:Complex_variables,
  author    = {Mark J. Ablowitz and Athanassios S. Fokas},
  publisher = {Cambridge University Press},
  title     = {Complex variables: introduction and applications},
  year      = {2003},
}

@Article{Stokes:On_the_discontinuity_of_arbitrary_constants_which_appear_in_divergent_developments,
  author  = {George G. Stokes},
  journal = {Trans. Cambridge Philos. Soc.},
  title   = {On the discontinuity of arbitrary constants which appear in divergent developments},
  year    = {1857},
  pages   = {105--128},
  volume  = {10},
}
